\documentclass[journal,onecolumn,11pt]{IEEEtran}
\usepackage{amsfonts,amsmath,amssymb}
\usepackage{indentfirst, setspace}
\usepackage{url,float}
\usepackage{xcolor}
\usepackage[a4paper,ignoreall]{geometry}
\usepackage{longtable}
\usepackage{graphicx}
\usepackage[a4paper,ignoreall]{geometry}
\usepackage{multicol}
\usepackage{stfloats}
\usepackage{enumerate}
\usepackage{cite}
\usepackage{amsthm}
\usepackage{booktabs}
\usepackage{mathrsfs}
\usepackage{multirow}
\usepackage{amssymb}
\usepackage{systeme}
\usepackage{verbatim}
\usepackage{cases}
\usepackage[square, comma, sort&compress, numbers]{natbib}
\usepackage[overload]{empheq}
\def\qu#1 {\fbox {\footnote {\ }}\ \footnotetext { From Qu: {\color{red}#1}}}
\def\hqu#1 {}
\def\kq#1 {\fbox {\footnote {\ }}\ \footnotetext { From KangQuan: {\color{blue}#1}}}
\def\hkq#1 {}

\usepackage{stfloats}

\newtheorem{Th}{Theorem}

\newtheorem{Prop}[Th]{Proposition}
\newtheorem{Prob}[Th]{Problem}
\newtheorem{Lemma}[Th]{Lemma}

\newtheorem{example}{Example}
\newtheorem{Conj}[Th]{Conjecture}
\newtheorem{Rem}[Th]{Remark}

\newcommand{\tr}{{\rm Tr}}

\newcommand{\gf}{{\mathbb F}}

\makeatletter
\newcommand{\figcaption}{\def\@captype{figure}\caption}
\newcommand{\tabcaption}{\def\@captype{table}\caption}
\makeatother

\begin{document}
	\title{Binary linear codes with few weights from two-to-one functions }
	\author{
		{ Kangquan Li, Chunlei Li, Tor Helleseth and Longjiang Qu}
		\thanks{\noindent Kangquan Li is with the College of Liberal Arts and Sciences,
			National University of Defense Technology, Changsha, 410073, China and is currently a visiting Ph.D. student at the Department of Informatics, University of Bergen, Bergen N-5020, Norway.
			Chunlei Li and Tor Helleseth are with the Department of Informatics, University of Bergen, Bergen N-5020, Norway.
			Longjiang Qu is with the College of Liberal Arts and Sciences,
			National University of Defense Technology, Changsha, 410073, China, and is also with
			the State Key Laboratory of Cryptology, Beijing, 100878, China. The work of Longjiang Qu was supported by the Nature Science Foundation of China (NSFC) under Grant  61722213, 11531002,   National Key R$\&$D Program of China (No.2017YFB0802000),  and the Open Foundation of State Key Laboratory of Cryptology. The work of Tor Helleseth and Chunlei Li was supported by the Research Council of Norway (No.~247742/O70 and No.~311646/O70).
			The work of Chunlei Li was also supported in part by the National Natural Science Foundation of China under Grant (No.~61771021). The work of Kangquan Li was supported by China Scholarship Council. Longjiang Qu is the corresponding author.
			
			\smallskip
			
			\textbf{Emails}: 			likangquan11@nudt.edu.cn,  chunlei.li@uib.no, tor.helleseth@uib.no, 
			ljqu\_happy@hotmail.com
		}
	}
	\maketitle{}

\begin{abstract}

In this paper, we apply two-to-one functions  over $\gf_{2^n}$  in two generic constructions of binary linear codes.
We consider two-to-one functions  in two forms: (1) generalized quadratic functions; and (2) $\left(x^{2^t}+x\right)^e$ with $\gcd(t, n)=1$ and $\gcd\left(e, 2^n-1\right)=1$.  
Based on the study of the Walsh transforms of those functions or their related-ones, we present many classes of linear codes with few nonzero weights, including one weight, three weights, four weights and five weights.
The weight distributions of the proposed codes with one weight and with three weights are determined. 
In addition, we discuss the minimum distance of the dual of the constructed codes and show that some of them achieve the sphere packing bound. { Moreover, several examples show that some of our codes are optimal and some have the best known parameters.}

\end{abstract}


\section{Introduction}

Let $q$ be a power of a prime $p,$ $\gf_{q}$ be the finite field of $q$ elements and  $\gf_{q}^*$ be its multiplicative group. 
An $[n, k, d]$ linear code $\mathcal{C}$ over $\gf_{q}$ is a $k$-dimensional subspace of $\gf_q^n$ with minimum (Hamming) distance $d$. 
{ 
	An $[n,k,d]$ code is sometimes usually said to be optimal when its minimum distance $d$ 
	achieves the maximum value 
	with respect to the Hamming bound \cite{huffman2010fundamentals}. Given an $[n, k, d]$ linear code $\mathcal{C}$ over $\gf_q$, 
	the dual code of $\mathcal{C}$ is an $[n, n-k]$ linear code defined by 
	$\mathcal{C}^{\perp} = \left\{ \mathbf{x}\in\gf_q^n: \mathbf{x}\cdot \mathbf{c} = 0, \forall \mathbf{c}\in\mathcal{C}  \right\},$  where $\mathbf{x}\cdot \mathbf{c} = \sum_{i=1}^{n}x_ic_i$ is a scalar product.
	}
 Let $A_i$ denote the number of codewords with Hamming weight $i$ in a code $\mathcal{C}$ of length $n$. The {weight enumerator} of $\mathcal{C}$ is defined by $1+A_1z+A_2z^2+\cdots+A_nz^n$. The sequence $(1,A_1,A_2,\cdots,A_n)$ is called the weight distribution of $\mathcal{C}$. A code $\mathcal{C}$ is said to be a $t$-weight code if the number of nonzero $A_i$ in the sequence $(A_1,A_2,\cdots,A_n)$ is equal to $t$.
Linear $t$-weight codes with small values of $t$
have many applications, including secret sharing schemes \cite{carlet2005linear,yuan2005secret}, authentication codes \cite{ding2005coding}, association schemes \cite{calderbank1984three}, strongly regular graphs \cite{calderbank1986geometry}, etc.  In particular, one-weight codes are known as constant-weight codes, which are closely connected to the theory of Steiner systems and designs.


Known linear codes with good properties are constructed largely by two generic approaches \cite{carlet2005linear,ding2007cyclotomic,ding2016}. 
By the first approach, linear codes over $\gf_q$ are defined based on a function $f$ from $\gf_{q^n}$ to itself 
as 
$$
\overline{\mathcal{C}}_f=\{
(\tr_n(af(x)+bx))_{x\in\gf_{q^n}},\,\, a,\,b \in \gf_{q^n}
\}
$$ or $$
\mathcal{C}_f=\{
(\tr_n(af(x)+bx))_{x\in\gf_{q^n}^*},\,\, a,\,b \in \gf_{q^n}
\}
$$ when $f(0)=0$,
where $n$ is a positive integer and $\tr_n$ is the trace function from $\gf_{q^n}$ to $\gf_{q}$. 
This generic construction has a long history and its importance is
supported by Delsarte’s Theorem \cite{Delsarte1975}. It also gives a coding-theory characterisation of APN monomials, almost bent functions, and bent functions \cite{carlet1998codes}, and 
cross-correlation between $m$-sequences and their decimations \cite{dobbertin2006niho}.
The second generic construction, introduced by Ding and Niederreiter, is described in terms of defining sets \cite{ding2007cyclotomic}. More specifically,
it takes a subset $D=\{d_1, d_2, \cdots, d_{\ell}\} \subseteq \gf_{q^n}$ and defines a linear code of length $\ell$ over $\gf_q$ as 
\begin{equation*}
\label{define set constructions}
\mathcal{C}_D = \left\{ \left( \tr_n(xd_1), \tr_n(xd_2), \cdots, \tr_n(xd_{\ell}) \right): x\in\gf_{q^n}  \right\},
\end{equation*}  where $D$ is called the defining set of the code $\mathcal{C}_D$. 
When the defining set $D$
is properly chosen, the code $\mathcal{C}_D$ can have good or optimal properties. This construction is generic in the sense that many classes of known codes could be produced
by selecting proper defining sets $D$. 
In recent years, researchers have proposed new  families of linear codes with few weights by considering defining sets derived from the support and image of certain functions over $\gf_{q^n}$, see \cite{ding2016,ding2015linear,ding2015class,heng2016three,zhou2016linear,tang2016linear} etc.
Interested readers may refer to a recent survey by Li and Mesnager in \cite{lirecent} and references therein for good or optimal linear codes constructed from these two generic approaches. 

Nonlinear functions over finite fields play important roles in cryptography, combinatorics, designs and sequence design. 
In coding theory, they have been employed in the above two generic constructions,  resulting in a number of linear codes with good or optimal properties.
Recently Mesnager and Qu in \cite{mesnager2019two} made a systematic study of two-to-one functions  over arbitrary finite fields, motivated by their close connection to special important primitives in symmetric cryptography.
Later, Li et al. further developed the study of two-to-one functions  over finite fields with characteristic $2$ and proposed some two-to-one trinomials and quadrinomials \cite{li2019further}. 

  {The research of linear codes from two-to-one functions, to our best knowledge, began in \cite{ding2015linear,ding2016construction}
  	when o-monomials and APN functions are discussed in the context. 
  In this paper we will provide a more comprehensive study of  two-to-one functions in constructing binary linear codes with few weights.  
 Two forms of two-to-one functions from $\gf_{2^n}$ to itself are considered.
  	The first form is the generalized quadratic polynomial $f(x)$, for which
	there exists a positive integer $e$ with $\gcd(e, 2^n-1)=1$ such that $f(x^e)$ is a quadratic function.
  	The second form is the function $(x^{2^t}+x)^e$ with $\gcd(t,n)=1$ and $\gcd(e, 2^n-1)=1$.  
Among the generalized quadratic polynomials, of particular interest are those
with large ranks  because they can produce linear codes with few weights.
Hence some two-to-one functions  in \cite{li2019further} and two newly constructed two-to-one polynomials are considered. 
As a result, we obtain many classes of  1-weight (a.k.a constant-weight), 3-weight, 5-weight binary linear codes by the two generic constructions.
  For the second form $(x^{2^t}+x)^e$, we give an interesting connection between the weight distribution of linear codes in second construction and the Walsh spectrum of the Boolean function $\tr_n(x^e)$.
 The connection enables us to derive many classes of 3-weight, 4-weight and 5-weight binary linear codes. 
}
By applying the Pless power moments, the weight distribution of the proposed constant-weight and 3-weight linear codes are determined.
We do not manage to determine the weight distribution of those 5-weight linear codes in this paper.   
{In the end, based on the experiment results, we propose some open problems for the linear codes.  
}


The remainder of this paper is organized as follows. Section 2 introduces mathematic foundations and auxiliary results.
Section 3 first recalls some known two-to-one functions  in \cite{li2019further} and then investigates the parameters of 
binary linear codes constructed from those two-to-one functions .
In Section 4, we construct two new classes of two-to-one functions  and propose 3-weight linear codes from them.  {In Section 5, we discuss the 
	properties of linear codes from the two-to-one functions of the form $(x^{2^t}+x)^e$.}
Finally, the concluding remark on our work is given in Section 6.

\section{Preliminaries}

This section presents basic notations, definitions and auxiliary results for the subsequent sections.
Throughout this paper, we will restrict our discussion to finite fields with characteristic $2$.

Let $n$ be a positive integer. For $m\mid n$, let $\tr_{n/m}(\cdot)$ denote the relative trace function from $\gf_{2^n}$ onto $\gf_{2^m}$, i.e., $\tr_{n/m}(x)=x+x^{2^m}+\cdots+x^{2^{\left(\frac{n}{m}-1\right)m}}$ for any $x\in\gf_{2^n}$. Particularly, when $m=1$, we use $\tr_n(\cdot)$ to denote the absolute trace function  from $\gf_{2^n}$ onto $\gf_2$. 
For any set $E$, we denote by $|E|$ the cardinality of $E$.  

\subsection{Binary codes from two-to-one functions}

Let $f(x)$ be a mapping from $\gf_{2^n}$ to itself with $f(0)=0$. Recall that the Walsh transform of $f(x)$ at $(a, b)\in \gf_{2^n}\times \gf_{2^n}$ is given by 
\begin{equation}
\label{Wfab} W_f(a,b) = \sum\limits_{x\in \gf_{2^n}}(-1)^{\tr_n(ax+bf(x))}.
\end{equation} Here we add the case that $b=0$ in the definition for convenience.

In the first generic construction,  the binary linear code from $f$ is given by
\begin{equation}\label{Eq-GenCons1}
\mathcal{C}_f=\left\{
\mathbf{c}_{a,b}=(\tr_n(ax+bf(x)))_{x\in\gf_{2^n}^*},\,\, a,\,b \in \gf_{2^n}
\right\}.
\end{equation} 
Note that with the restriction $f(0)=0$, we have $\tr_n(ax+bf(x))=0$ for any $a,\,b \in \gf_{2^n}$ when $x=0$.
Hence the code $\mathcal{C}_f$ is commonly considered in the literature over the code $\overline{\mathcal{C}}_f$ in the first generic construction.
In this case the linear code $\mathcal{C}_f$ has length $2^n-1$ and dimension at most $2n$. Furthermore, in order to determine the dimension of $\mathcal{C}_{f}$, we only need to compute the number of $a, b\in \gf_{2^n}$ such that
$\tr_n(ax+bf(x))=0$ for any $x\in\gf_{2^n}$ since the code is linear. Equivalently, the dimension of $\mathcal{C}_{f}$ is equal to $2n-d_{K_1}$, where $d_{K_1}$ is the dimension of the 
$\gf_2$-vector space
\begin{equation}\label{dimensionK1}
K_1 = \left\{ (a, b)\in \gf_{2^n}^2:  \sum_{x\in \gf_{2^n}}(-1)^{\tr_n(ax+bf(x))}=2^n  \right\}.
\end{equation}
For a codeword $\textbf{c}_{a,b}$, its Hamming weight is given by 
\begin{equation}\label{WtDist-Relation1}
\begin{split}
\mathrm{wt}(\mathbf{c}_{a,b}) 
	& = 2^n-1-\#\{x\in \gf_{2^n}^*\,:\, \tr_n(ax+bf(x))=0\} \\
&=2^{n-1}-\dfrac{1}{2} \sum\limits_{x\in\gf_{2^n}}(-1)^{\tr_n(ax+bf(x))}\\
&=2^{n-1}-\dfrac{1}{2} W_f(a,b).
\end{split}
\end{equation}
Therefore,  the weight distribution of $\mathcal{C}$ can be directly derived from the Walsh spectrum of $f(x)$:
$$
\left\{ \left\{
W_f(a,b)=\sum\limits_{x\in\gf_{2^n}}(-1)^{\tr_n(ax+bf(x))}\,:\, a, \, b\in \gf_{2^n}
\right\} 
\right\}.
$$ Namely, if a value of $W_f(a,b)$ occurs $X$ times in the Walsh spectrum of $f$, then there are $X/2^{d_{K_1}}$ codewords 
in $\mathcal{C}_f$ with Hamming weight $2^{n-1}-\dfrac{1}{2} W_f(a,b)$.
In particular, when the Walsh transforms of $f$ take only three values $v_1,v_2,v_3$ for all $a\in \gf_{2^n}$ and $b\in \gf_{2^n}^*$, 
the value distribution of $W_f(a,b)$ can be calculated by solving the following equations derived from the first three power moment identities: 
\begin{equation}\label{Eq-WtDist1}
\begin{cases}
\sum\limits_{a\in \gf_{2^n}} \sum\limits_{b\in \gf_{2^n}} W_f^0(a,b) & =2^{2n} = 2^{d_{K_1}}+X_1+X_2+X_3, \\
\sum\limits_{a\in \gf_{2^n}} \sum\limits_{b\in \gf_{2^n}} W_f(a,b) & =2^{2n} = 2^{d_{K_1}}\cdot 2^n+v_1X_1+v_2X_2+v_3X_3, \\
\sum\limits_{a\in \gf_{2^n}} \sum\limits_{b\in \gf_{2^n}} W_f^2(a,b) & =2^{3n}= 2^{d_{K_1}}\cdot 2^{2n}+v_1^2X_1+v_2^2X_2+v_3^2X_3, \\
\end{cases}
\end{equation}
where $X_i$ is the occurrences of $W_f(a,b)=v_i$'s, $i=1, 2, 3$ in the Walsh spectrum of $f$. Then the weight distribution of $\mathcal{C}_f$
can be determined accordingly.

\medskip

In the second construction, let $D(f)=\{ f(x): x \in\gf_{2^n} \}\backslash\{0\}=\{d_1,d_2,\ldots, d_{\ell}\}$ and define the binary linear code
\begin{equation}\label{Eq-GenCons2}
\mathcal{C}_{D(f)}= \left\{ \textbf{c}_b=\left( \tr_n(bd_1), \tr_n(bd_2), \ldots, \tr_n(bd_{\ell}) \right): b\in\gf_{2^n}  \right\}.
\end{equation} 
It is clear that the code $\mathcal{C}_{D(f)}$ has length $\ell=|D(f)|$ and dimension at most $n$. 
Furthermore, in order to determine the dimension of $\mathcal{C}_{D(f)}$, we need to compute the number of $b\in \gf_{2^n}$ such that
$\tr_n(bf(x))=0$ for any $x\in\gf_{2^n}$ since the code is linear. Equivalently, the dimension of $\mathcal{C}_{D(f)}$ is equal to $n-d_{K_2}$, where $d_{K_1}$ is the dimension of the 
$\gf_2$-vector space
\begin{equation}\label{dimension}
K_2 = \left\{ b\in\gf_{2^n} \,:\, \sum_{x\in \gf_{2^n}}(-1)^{\tr_n(bf(x))}=2^n  \right\}.
\end{equation}
For any $b\in\gf_{2^n}$, the Hamming weight of a codeword $\mathbf{c}_b$ in $\mathcal{C}_{D(f)}$  is given by
\begin{equation*}
\begin{split}
\mathrm{wt}(\mathbf{c}_{b})   =  |\{1\leq i \leq \ell \,:\, \tr_n(bd_i)=1\}|  
= \dfrac{1}{2}\left(|D_f| - \sum\limits_{d\in D(f)} (-1)^{\tr_n(bd)}\right).
\end{split}
\end{equation*}
From the above formula, one sees that the weight distribution of the linear code $\mathcal{C}_{D(f)}$ is 
essentially the value distribution of a partial exponential sum, which is generally intractable if $f$ is not properly chosen.

Suppose that $f(x)$ is a two-to-one mapping from $\gf_{2^n}$ to itself, which means that $|f^{-1}(a)|=2$ for any $a\in \mathrm{Im}(f)$. 
Then the linear code $\mathcal{C}_{D(f)}$  has length $|D_f|=2^{n-1}-1$ and the Hamming weight of its codeword is given by
\begin{equation}
\label{wt_new}
\mathrm{wt}(\mathbf{c}_b) = \frac{1}{2}\left( |D_f|-  \frac{1}{2}\sum_{x\in \gf_{2^n}}(-1)^{\tr_n(bf(x))} +1\right) = 2^{n-2} - \frac{1}{4}\sum_{x\in \gf_{2^n}}(-1)^{\tr_n(bf(x))}.
\end{equation}
From \eqref{dimension} and \eqref{wt_new}, one sees that the dimension and the weight distribution of $\mathcal{C}_{D(f)}$ heavily depend on the value of 
\begin{equation}
\label{wf}
W_f(b) \triangleq W_f(0,b)=\sum_{x\in \gf_{2^n}}(-1)^{\tr_n(bf(x))}, \quad b\in \gf_{2^n}.
\end{equation} In particular, if $W_f(b)$ takes only three values $v_1, v_2$ and $v_3$ for $b\in \gf_{2^n}^{*}$,
then the code $\mathcal{C}_{D(f)}$ has three nonzero weights, namely,  $\mathrm{w}_i=2^{n-2}-v_i/4$ for $i=1, 2, 3$. 
Denote by $A_{i}$ the number of codewords with weight $\mathrm{w}_i$ in $\mathcal{C}_{D(f)}$.
Note that the dual of $\mathcal{C}_{D(f)}$ has Hamming weight no less than $3$ since $\tr_n(xd)=0$ holds for all $x\in \gf_{2^n}$ if and only if $d=0$.
The first three Pless Power Moments \cite[p. 260]{huffman2010fundamentals} leads to the following system of equations:
\begin{equation} \label{Eq-WtDist2}
\begin{cases} A_{1} +A_{2}+A_{3}= 2^n-1 \\ 
\mathrm{w}_1 A_{1} + \mathrm{w}_2 A_{2}+ \mathrm{w}_3 A_{3}= \ell 2^{n-1} \\ 
\mathrm{w}_1^2 A_{1} + \mathrm{w}_2^2 A_{2} + \mathrm{w}_3^2 A_{3}= \ell ( \ell +1) 2^{n-2}, \end{cases}
\end{equation} where $\ell = 2^{n-1}-1$. Therefore, the weight distribution of  $\mathcal{C}_{D(f)}$ can be determined 
from the above system of equations when it is shown to have only three nonzero weights.

The above discussion shows that for a two-to-one mapping $f$, 
the parameters of the linear codes $\mathcal{C}_f$ in \eqref{Eq-GenCons1} and $\mathcal{C}_{D(f)}$ in \eqref{Eq-GenCons2}
depend on the investigation of the Walsh transform of $f$. In addition, it is clear that the number of nonzero weights in $\mathcal{C}_{D(f)}$ 
is no more than that of $\mathcal{C}_f$. {Therefore, we will focus on the two-to-one functions of which the Walsh transforms have few different values.} 

{At the end of this subsection, we consider the minimum distance of dual codes of $\mathcal{C}_f$ in \eqref{Eq-GenCons1}  and $\mathcal{C}_{D(f)}$ in \eqref{Eq-GenCons2} in the following theorem.

\begin{Th}
		\label{dual_codes}
	Let $f$ be a two-to-one mapping over $\gf_{2^n}$, $\mathcal{C}_f$ and $\mathcal{C}_{D(f)}$ be defined as in \eqref{Eq-GenCons1} and \eqref{Eq-GenCons2}, respectively. Let $\mathcal{C}_f^{\perp}$ and $\mathcal{C}_{D(f)}^{\perp}$ be the dual code of $\mathcal{C}_f$ and $\mathcal{C}_{D(f)}$, respectively. Let $d_{K_1}$ and $d_{K_2}$ be defined as in \eqref{dimensionK1} and \eqref{dimension}, respectively. Then 
	\begin{enumerate}[(1)]
		\item $\mathcal{C}_f^{\perp}$ is a $\left[2^n-1, 2^n-1-2n+d_{K_1}\right]$ binary code with the minimum distance $d_f^{\perp}$ satisfying $ 3 \le d_f^{\perp}\le 6$. Particularly, if $d_{K_1}\ge 2$, $3\le d_f^{\perp}\le 4$. Moreover, $d_f^{\perp}=3$ if and only if there exist two distinct elements $x_1,x_2\in\gf_{2^n}^{*}$ such that $f(x_1)+f(x_2)+f(x_1+x_2)=0$. 
		
		\item $\mathcal{C}_{D(f)}^{\perp}$ is a $\left[2^{n-1}-1, 2^{n-1}-1-n+d_{K_2}\right]$ binary code with the minimum distance $d_{D(f)}^{\perp}$ satisfying $3\le d_{D(f)}^{\perp} \le 4$. Particularly, when $d_{K_2}=1$, the equality of the sphere packing bound can be achieved.  Moreover, $d_{D(f)}^{\perp}=3$ if and only if there exist three distinct elements $x_1,x_2,x_3\in\gf_{2^n}^{*}$ such that $f(x_i)\neq f(x_j)$ for $1\le i <j\le 3$ and $f(x_1)+f(x_2)+f(x_3)=0$. 
	\end{enumerate}
\end{Th}

\begin{proof}
 From the above discussion, we know that the linear code $\mathcal{C}_f$ (resp. $\mathcal{C}_{D(f)}$) has length $2^n-1$ (resp. $2^{n-1}-1$) and dimension $2n-d_{K_1}$ (resp. $n-d_{K_2}$). Then according to the definition,	the length and dimension of $\mathcal{C}_f^{\perp}$ and $\mathcal{C}_{D(f)}^{\perp}$ can be trivially determined. Thus it suffices to consider the minimum distances of the dual codes. 
	
	(1)  For $\mathcal{C}_f^{\perp}$, it is clear that $d_f^{\perp}\neq1$. If $\mathcal{C}_f^{\perp}$ has a codeword of weight two, then there exist two distinct elements $x_1,x_2\in\gf_{2^n}^{*}$ such that $x_1+x_2=0$ and $f(x_1)+f(x_2)=0$, which is impossible. Thus $d_f^{\perp}\ge 3$.  We now prove that $d_f^{\perp}\le 6$. Suppose that $d_f^{\perp}\ge 7$. We would then have 
	\begin{eqnarray*}
& &	\sum_{i=0}^3 \binom{2^n-1}{i}(2-1)^i \\
&=& 1 + 2^n-1 + \left(2^n-1\right)\left(2^{n-1}-1\right) + \frac{\left(2^n-1\right)\left(2^n-3\right)\left(2^{n-1}-1\right)}{3} \\
&=& 2^{2n-1}-2^{n-1}+1 + \frac{\left(2^n-1\right)\left(2^n-3\right)\left(2^{n-1}-1\right)}{3}  \\
&>& 2^{2n-d_{K_1}}, 
	\end{eqnarray*}
which is contrary to the sphere packing bound. Thus $d_f^{\perp}\le 6$. In particular, obviously, if $d_{K_1}\ge 2$, $d_f^{\perp}\le 4$. Moreover, according to the definition, $\mathcal{C}_f^{\perp}$ has a codeword of weight three if and only if there are pairwise distinct elements $x_1,x_2,x_3\in\gf_{2^n}^{*}$ such that 
 $$\begin{cases} x_1+x_2+x_3 = 0 \\ 
f(x_1)+f(x_2)+f(x_3) = 0, \end{cases}$$
i.e., there exist two distinct elements $x_1,x_2\in\gf_{2^n}^{*}$ such that $f(x_1)+f(x_2)+f(x_1+x_2)=0$. 

(2) For $\mathcal{C}_{D(f)}^{\perp}$, it is also clear that $d_{D(f)}^{\perp}\neq1$. Moreover, if $\mathcal{C}_{D(f)}^{\perp}$ has a codeword of weight two, then there exist two distinct elements $d_1,d_2\in \mathrm{Im}(f)$ such that $d_1+d_2=0$, which is contrary.  Thus $d_{D(f)}^{\perp}\ge3$. We next show that $d_f^{\perp}\le 4$. Suppose that $d_f^{\perp}\ge 5$. We would then have 
\begin{eqnarray*}
& &	\sum_{i=0}^2 \binom{2^{n-1}-1}{i}(2-1)^i \\
&=& 1+2^{n-1}-1 + \left(2^{n-1}-1\right)\left(2^{n-2}-1\right) \\
&=& 2^{2n-3}-2^{n-2}+1 \\
&>& 2^{n-d_{K_2}},
\end{eqnarray*}
which is contrary to the sphere packing bound. Thus $d_f^{\perp}\le 4$. Particularly, when $d_{K_2}=1$, the equality of the sphere packing bound can be achieved, namely, $$\sum_{i=0}^1\binom{2^{n-1}-1}{i}(2-1)^i = 2^{n-1} = 2^{n-d_{K_2}}. $$
Moreover, from the definition, $\mathcal{C}_{D(f)}^{\perp}$ has a codeword of weight three if and only if there are three distinct elements $d_1,d_2,d_3\in\mathrm{Im}(f)$ such that $d_1+d_2+d_3=0$, i.e., there exist three distinct elements $x_1,x_2,x_3\in\gf_{2^n}^{*}$ such that $f(x_i)\neq f(x_j)$ for $1\le i <j\le 3$ and $f(x_1)+f(x_2)+f(x_3)=0$. 
\end{proof}
}

 We need to recall some useful results on the Walsh transforms of quadratic functions.

\subsection{Quadratic functions and Walsh transforms}

Let $Q(x)$ be a quadratic function from $\gf_{2^n}$ to itself, i.e., it has algebraic degree $2$. Let $\varphi(x)=\tr_n(Q(x))$ and define its
associated bilinear mapping as
$$
B_{\varphi}(x,y)=\varphi(x+y)+\varphi(x)+\varphi(y).
$$ The set 
$$
V_{\varphi} = \{y \in \gf_{2^n}\,:\, \varphi(x+y)+\varphi(x)+\varphi(y)=0 \text{ for any } x \in \gf_{2^n} \}
$$
is an $\gf_{2}$-vector space and is known as the kernel of the bilinear mapping $B_{\varphi}(x,y)$. 
The rank of $\varphi(x)$ is defined by 
$$
\text{Rank}(\varphi(x)) = n - \dim_{\gf_2}(V_{\varphi}).
$$
Observe that 
\begin{equation}\label{Eq-Quad0}
\begin{split}
\left(\sum\limits_{x\in\gf_{2^n}}(-1)^{\tr_n(Q(x))}\right)^2 &=
\sum\limits_{x\in \gf_{2^n}}(-1)^{\tr_n(Q(x))} \sum\limits_{y\in \gf_{2^n}}(-1)^{\tr_n(Q(y))}
\\&=\sum\limits_{x,y\in \gf_{2^n}}(-1)^{\tr_n(Q(x+y)+Q(y))}
\\&=\sum\limits_{y\in \gf_{2^n}}(-1)^{\tr_n(Q(y))}\sum\limits_{x\in \gf_{2^n}}(-1)^{\tr_n(Q(x+y)+Q(x)+Q(y))}
\\&= 2^n \sum\limits_{y\in V_{\varphi}}(-1)^{\tr_n(Q(y))}. 
\end{split}
\end{equation} 
By the definition of the kernel $V_{\varphi}$, it is readily seen that the function $\varphi(y)=\tr_n(Q(y))$ is linear over $V_{\varphi}$. 
Then one has 
\begin{equation}\label{Eq-Quad}
\sum\limits_{x\in\gf_{2^n}}(-1)^{\tr_n(Q(x))} = \begin{cases}
\pm 2^{\frac{n+d}{2}}, & \text{if } \varphi(x)=0 \text{ for all } x \in V_{\varphi}, \\
0, & \text{otherwise,} 
\end{cases}
\end{equation} where $d$ is the dimension of $V_{\varphi}$ over $\gf_2$.

For a quadratic function $f$ from $\gf_{2^n}$ to itself, the bilinear mapping of  $\tr_n(bf(x)+ax)$ is the same as that of $\tr_n(bf(x))$ for any nonzero $b$ in $\gf_{2^n}$.
Therefore, the Walsh transform of $f$ at $(a,b)$ can be given, similar to \eqref{Eq-Quad},  as below:
\begin{equation}\label{Eq-Quad-2}
W_f(a,b) =\sum\limits_{x\in\gf_{2^n}}(-1)^{\tr_n(ax+bf(x))} = \begin{cases}
\pm 2^{\frac{n+d_{b}}{2}}, & \text{if } \tr_n(bf(x)+ax)=0 \text{ for all } x \in V_{\varphi_b}, \\
0, & \text{otherwise.} 
\end{cases}
\end{equation}
where $d_b$ is the dimension of $V_{\varphi_b}$ over $\gf_{2}$ and $V_{\varphi_b}$ is the kernel of the bilinear mapping of $\varphi_b(x)=\tr_n(bf(x))$.

In Sections 3 and 4, we will discuss the properties of linear codes defined in \eqref{Eq-GenCons1} and \eqref{Eq-GenCons2} from generalized quadratic functions.
The Walsh transform of quadratic functions $f(x)$ in \eqref{Eq-Quad-2} will be heavily used in the discussion.

\subsection{Factorization of low-degree polynomials}

The following lemma describes the factorization of a cubic polynomial over $\gf_{2^m}$. If $f$ factors over $\gf_{2^m}$ as a product of three linear factors we write $f=(1,1,1)$, if $f$ factors as a product of a linear factor and an irreducible quadratic factor we write $f=(1,2)$ and finally if $f$ is irreducible over $\gf_{2^m}$ we write $f=(3)$.

\begin{Lemma}
	\label{cubic}
	\cite{williams1975note}
	Let $f(x)=x^3+ax+b\in\gf_{2^m}[x]$ and $b\neq0$. Let $t_1,t_2$ denote the solutions of $t^2+bt+a^3=0$. Then the factorizations of $f(x)$ over $\gf_{2^m}$ are characterized as follows:
	\begin{enumerate}[(1)]
		\item $f=(1,1,1)$ if and only if $\tr_m\left(\frac{a^3}{b^2}\right)=\tr_m(1)$, $t_1,t_2$ cubes in $\gf_{2^m}$ ($m$ even), $\gf_{2^{2m}}$ (m odd);
		\item $f=(1,2)$ if and only if $\tr_m\left(\frac{a^3}{b^2}\right)\neq\tr_m(1)$;
		\item $f=(3)$ if and only if $\tr_m\left(\frac{a^3}{b^2}\right)=\tr_m(1)$, $t_1,t_2$ not cubes in $\gf_{2^m}$ ($m$ even), $\gf_{2^{2m}}$ (m odd).
	\end{enumerate}
\end{Lemma}

\begin{Lemma}
	\label{cubic_solution}
	\cite{williams1975note}
	Let $f(x)=x^3+ax+b\in\gf_{2^m}[x]$ and $b\neq0$. Let $t$ be one solution  of $t^2+bt+a^3=0$ and $\epsilon$ be one solution of $x^3=t$. Then $r=\epsilon+\frac{a}{\epsilon}$ is a solution of $f(x)=0$.
\end{Lemma}

The following lemma characterizes the factorization of a quartic polynomial over $\gf_{2^n}.$
\begin{Lemma}
	\cite{leonard1972quartics}
	\label{quartic}
	Let $f(x)=x^4+a_2x^2+a_1x+a_0$ with $a_i\in\gf_{2^n}$ and $a_0a_1\neq0$. Let $f_1(y)=y^3+a_2y+a_1$ and $r_1,r_2,r_3$ denote roots of $f_1(y)=0$ when they exist in $\gf_{2^n}$. Set $w_i=a_0\frac{r_i^2}{a_1^2}$. Then the factorization of $f(x)$ over $\gf_{2^n}$ is characterized as follows:
	\begin{enumerate}[(1)]
		\item $f=(1,1,1,1)$ if and only if $f_1=(1,1,1)$ and $\tr_n(w_1)=\tr_n(w_2)=\tr_n(w_3)=0$;
		\item $f=(2,2)$ if and only if $f_1=(1,1,1)$ and $\tr_n(w_1)=0$, $\tr_n(w_2)=\tr_n(w_3)=1$;
		\item $f=(1,3)$ if and only if $f_1=(3)$;
		\item $f=(1,1,2)$ if and only if $f_1=(1,2)$ and $\tr_n(w_1)=0$;
		\item $f=(4)$ if and only if $f_1=(1,2)$ and $\tr_n(w_1)=1$.
	\end{enumerate} 
\end{Lemma}

\section{Binary linear codes from known two-to-one trinomials and quadrinomials}

In this section, we will propose several binary codes with few weights, which are constructed from known two-to-one functions .
We first recall some two-to-one functions  recently obtained in \cite{li2019further}.

\begin{Lemma}
	\cite{li2019further}
	\label{2-1-1}
	Let $f(x)=x^{\frac{2^{n-1}+2^m-1}{3}}+x^{2^m}+\omega x \in\gf_{2^n}[x]$, where $n=2m$, $m$ is odd and $\omega\in\gf_{2^2}\backslash\gf_2$. Then $f(x)$ is two-to-one over $\gf_{2^n}$.
\end{Lemma} 

\begin{Lemma}
	\label{2-1-2}
	\cite{li2019further}
	Let $n=2m+1$. Then the following quadrinomials are all two-to-one over $\gf_{2^n}$:
	\begin{enumerate}[(1)]
		\item $f(x)=x^{2^{m+1}+2}+x^{2^{m+1}}+x^2+x$;
		\item $f(x)=x^{2^{m+1}+2}+x^{2^{m+1}+1}+x^2+x$;
		\item $f(x)=x^{2^{m+2}+4}+x^{2^{m+1}+2}+x^2+x$;
		\item $f(x)=x^{2^n-2^{m+1}+2}+x^{2^{m+1}}+x^2+x$.
	\end{enumerate}
\end{Lemma}

\begin{Lemma}
	\label{2-1-3}
	\cite{li2019further}
	Let $n=3m$. Then the following quadrinomials are two-to-one over $\gf_{2^n}$:
	\begin{enumerate}[(1)]
		\item $f(x)=x^{2^{2m}+1}+x^{2^{m+1}}+x^{2^m+1}+x$ with $m\not\equiv 1\pmod 3$;
		\item $f(x)=x^{2^{2m}+2^m}+x^{2^{2m}+1}+x^{2^m+1}+x$.
	\end{enumerate}
\end{Lemma}

Below we shall investigate the parameters of the constructed linear codes $\mathcal{C}_f$ and $\mathcal{C}_{D(f)}$.
According to different forms of $n$ in Lemmas \ref{2-1-1} - \ref{2-1-3}, we divide them into three subsections.

\subsection{The case $n=2m$}

The following binary linear code is derived from the two-to-one polynomial in Lemma \ref{2-1-1}.

\begin{table}[!t]
	\centering
	\caption{The weight distribution of the codes $\mathcal{C}_{f}$ in Theorem \ref{Th_n=2m}}	\label{Th_n=2m_table1}
	\begin{tabular}{cc}
		\hline
		Weight & Multiplicity \\
		\hline
		$0$ & $1$\\
		$2^{n-1}-2^{m}$ & $2^{4m-3}+2^{3m-2}-2^{2m-3}-2^{m-2}$  \\
		$2^{n-1}$ & $3\cdot2^{4m-2} +2^{2m-2} - 1 $ \\
		$2^{n-1}+2^{m}$ & $2^{4m-3} + 2^{m-2}- 2^{3m-2} - 2^{2m-3}$\\
		\hline
	\end{tabular}
\end{table}

\begin{table}[!t]
	\centering
	\caption{The weight distribution of the codes $\mathcal{C}_{D(f)}$ in Theorem \ref{Th_n=2m}}	\label{Th_n=2m_table2}
	\begin{tabular}{cc}
		\hline
		Weight & Multiplicity \\
		\hline
		$0$ & $1$\\
		$2^{n-2}-2^{m-1}$ & $2^{n-3}+2^{{m-2}}$  \\
		$2^{n-2}$ & $3\cdot 2^{n-2}-1$ \\
		$2^{n-2}+2^{m-1}$ & $2^{n-3}-2^{{m-2}}$\\
		\hline
	\end{tabular}
\end{table}

\begin{Th}
	\label{Th_n=2m}
	Let $f(x)=x^{\frac{2^{n-1}+2^m-1}{3}}+x^{2^m}+\omega x \in\gf_{2^n}[x]$ with $n=2m$, where $m>1$ is odd and $\omega\in\gf_{2^2}\backslash\gf_2$. 
	Define two linear codes $\mathcal{C}_{f}$ and $\mathcal{C}_{D(f)}$ as in \eqref{Eq-GenCons1} and \eqref{Eq-GenCons2}, respectively. 
	Then, 
	\begin{enumerate}[(1)]
		\item  $\mathcal{C}_{f}$ is a $\left[2^{n}-1, 2n\right]$ binary linear code with weight distribution in Table \ref{Th_n=2m_table1}.
		\item  $\mathcal{C}_{D(f)}$ is a $\left[2^{n-1}-1, n\right]$ binary linear code with weight distribution in Table \ref{Th_n=2m_table2}.
	\end{enumerate}
\end{Th}
\begin{proof}
We first compute the value of $W_f(a,b)$, i.e., the Walsh transform of $f$, and $W_f(b)$ defined as  in \eqref{wf}, for any $a,b\in\gf_{2^n}$.	
It is obvious that $W_f(a,b)=2^n$ when $a=b=0$.  Below we consider the cases where $(a,b)\neq(0,0)$.

By the Euclidean algorithm we have
	$$(2^{2m}-1)\times (2^{m+1}+5) = (2^{2m-1}+2^m-1)\times (2^{m+2}+2) - 3,$$
	which implies $\gcd\left( \frac{2^{n-1}+2^m-1}{3}, 2^{2m}-1 \right) = 1$. 
	Define
	$$
	f_1(x)=f\left(x^{2^{m+2}+2}\right) =  x+x^{2^{m+1}+4}+\omega x^{2^{m+2}+2}
	$$ 
	and 
	$$Q(x)=ax^{2^{m+2}+2}+bf_1(x) = bx+ bx^{2^{m+1}+4}+ (b\omega+a) x^{2^{m+2}+2}.$$
	Then the Walsh transform
	\begin{eqnarray*}
		W_f(a,b) = \sum_{x\in \gf_{2^n}}(-1)^{\tr_n\left(ax^{2^{m+2}+2} + bf\left(x^{2^{m+2}+2}\right)\right)} 
	 =  \sum_{x\in \gf_{2^n}}(-1)^{\tr_n\left(Q(x)\right)}.
	\end{eqnarray*}
Note that the bilinear form of $\varphi_{a,b}(x)=\tr_n(Q(x))$ is given by
\begin{equation*}
\begin{split}
B_{\varphi_{a,b}}(x,y) & = \varphi_{a,b}(x+y)+\varphi_{a,b}(x)+\varphi_{a,b}(y) 
\\ &= \tr_n\left( by^4x^{2^{m+1}}+ by^{2^{m+1}}x^4 + (b\omega+a) y^2x^{2^{m+2}} + (b\omega+a) y^{2^{m+2}}x^2 \right)
\\ & = \tr_n\left( L_{a,b}(y) x^{2^{m+2}} \right)
\end{split} 
\end{equation*}	with 
$$L_{a,b}(y) = \Delta^2y^8 + \Delta^{2^m}y^2, ~~~ \text{and}~~~ \Delta = b^{2^{m}}\omega^2 + b + a^{2^{m}}.$$
Let $\ker (L_{a,b}) = \left\{ y \in\gf_{2^n}: L_{a,b}(y) =0  \right\}.$ 
From (\ref{Eq-Quad}), we have 
$$ W_f(a,b)  = 2^n\sum_{y\in \ker(L_{a,b})} (-1)^{\varphi_{a,b}(y)}
 =  \begin{cases}
\pm 2^{\frac{n+2}{2}}, & \text{if } \varphi_{a,b}(x)=0 \text{ for all } x \in \ker(L_{a,b}), \\
0, & \text{otherwise.} 
\end{cases}$$

Now we discuss the values of $\varphi_{a,b}(x)$ on the kernel $\ker (L_{a,b})$. 
When $\Delta=0$, we have $L_{a,b}(y)=0$ for any $y\in \gf_{2^n}$. 
When $\Delta\neq 0$,   by computation we have 
$$\ker (L_{a,b}) = \left\{ 0, ~y_0, ~y_0\omega, ~y_0 \omega^2  \right\},~~\text{where} ~~ y_0 = \Delta^{\frac{2^{m-1}-1}{3}}.$$ 
Moreover, 
	\begin{eqnarray*}
		\varphi_{a,b}(y_0) &=& \tr_n\left( by_0 + \frac{b}{b^{2^m}\omega^2+b+a^{2^m}} + \frac{\omega b + a}{b\omega+b^{2^m}+ a}  \right) \\
		&=& \tr_n\left( by_0 + \frac{b+b^{2^m}\omega^2 +a^{2^m}}{b^{2^m}\omega^2+b+a^{2^m}} \right) = \tr_n\left(by_0\right).
	\end{eqnarray*}
	Similarly, we have 
	$$\varphi_{a,b}(\omega y_0)=\tr_n\left(\omega by_0\right) \text{ and } \varphi_{a,b}(\omega^2 y_0)=\tr_n\left(\omega^2 by_0\right)=\varphi_{a,b}(y_0)+\varphi_{a,b}(\omega y_0).$$ 
	
In the following, we assume $a=0$ and show that there exist some $b$'s $\in \gf_{2^n}$ such that $\varphi_{0,b}(x)=0 \text{ for all } x \in \ker(L_{a,b})$ and $\varphi_{0,b}(x)= 1 \text{ for some } x \in \ker(L_{a,b})$, which implies
$$W_f(b) \in\left\{ 0, \pm  2^{\frac{n+2}{2}} \right\}.$$

 It is well known that for any elements $b\in\gf_{2^n}$, there exist unique $b_1,b_2\in\gf_{2^m}$ such that $b=b_1+b_2 \omega $ since $m$ is odd.   Plugging $b=b_1+\omega b_2$ into the expression of $y_0$, we get 
	$$y_0=\left( b^{2^m}\omega^2 + b \right)^{\frac{2^{m-1}-1}{3}} = \left( (b_1+b_2\omega^2)\omega^2 + b_1+b_2 \omega  \right)^{\frac{2^{m-1}-1}{3}} = (b_1\omega)^{\frac{2^{m-1}-1}{3}}.$$
	Let $g(b)=by_0 = b\left( b^{2^m}\omega^2 + b \right)^{\frac{2^{m-1}-1}{3}}$. Then $$g(b)=(b_1+b_2\omega) (b_1\omega)^{\frac{2^{m-1}-1}{3}} = \omega^{\frac{2^{m-1}-1}{3}}\left( b_1^{\frac{2^{m-1}+2}{3}} + b_1^{\frac{2^{m-1}-1}{3}}b_2\omega \right).$$
	If there exist two elements $\hat{b}=\hat{b}_1+\hat{b}_2 \omega,\widetilde{b}= \widetilde{b}_1+\widetilde{b}_2\omega$ with $\hat{b}_1,\hat{b}_2,\widetilde{b}_1,\widetilde{b}_2\in\gf_{2^m}$ such that $g(\hat{b})=g(\widetilde{b})$, then we have $$ \hat{b}_1^{\frac{2^{m-1}+2}{3}} = \widetilde{b}_1^{\frac{2^{m-1}+2}{3}} ~~\text{and}~~ \hat{b}_1^{\frac{2^{m-1}-1}{3}}\hat{b}_2 = \widetilde{b}_1^{\frac{2^{m-1}-1}{3}}\widetilde{b}_2.  $$ 
	Since $\gcd\left(\frac{2^{m-1}+2}{3}, 2^m-1  \right) = 1$, we have $\hat{b}_1=\widetilde{b}_1$. 
	Then if $\hat{b}_1=\widetilde{b}_1\neq0$, $\hat{b}_2=\widetilde{b}_2$; if $\hat{b}_1=\widetilde{b}_1=0$, $\hat{b}_2,\widetilde{b}_2\in\gf_{2^m}$. Thus 
	for $b=b_1+\omega b_2$, if $b_1=0$, $g(b)=0$; if $b_1\neq0$, $g(b)$ is bijective. Let $\mathrm{Im}(g)=\left\{ g(b): b \in \gf_{2^n} \right\}$. Then 
	$$|\mathrm{Im}(g)| = 2^n-2^m+1.$$ 
	According to the property of the trace function, there exist $(2^n-2^{n-2})$ elements $x\in\gf_{2^n}$ such that $\tr_{n/2}(x) \neq 0$. Since $2^n-2^m+1>2^n-2^{n-2}$, there must exist some $b\in\gf_{2^n}$ such that $\tr_{n/2}(g(b))=0$  and then $$\tr_n\left(by_0\right) = \tr_2\left(\tr_{n/2}(g(b))\right) = 0$$ and $$\tr_n\left(\omega by_0\right)= \tr_2\left( \tr_{n/2}(\omega g(b)) \right) = \tr_2\left(\omega\tr_{n/2}(g(b))\right)=0.$$ Clearly, there also exist some $b\in\gf_{2^n}$ such that $\tr_{n/2}(g(b))=1$ and then $\tr_n\left(\omega by_0\right) = \tr_2(\omega)=1$.  Thus 
	$W_f(b) \in\{ 0, \pm 2^{\frac{n+2}{2}} \}$ and obviously, $W_f(a,b)\in\{ 0, \pm 2^{\frac{n+2}{2}} \}$ for $(a,b)\neq(0,0)$. 
	
	With the analysis of possible values of $W_f(a,b)$ and $W_f(b)$, we're now ready to determine the parameters of $\mathcal{C}_f$ and $\mathcal{C}_{D(f)}$ in the following. 
	
	(1) For the linear code $\mathcal{C}_f$, since $W_f(a,b)=0$ if and only if $a=b=0$,  it follows from (\ref{dimensionK1}) that the dimension is $2n$. Moreover, for any $a,b\in\gf_{2^n}$, $W_f(a,b)\in\left\{ 0, 2^n, \pm 2^{\frac{n+2}{2}}  \right\}.$  
	Let 	$${v}_1 = -2^{\frac{n+2}{2}},~~ {v}_2=0,~~ {v}_3 = 2^{\frac{n+2}{2}}.$$
	According to (\ref{Eq-WtDist1}), the occurrences of $W_f(a,b)=v_i$'s, $i=1,2,3$ in the Walsh spectrum of $f$ are 
	$$ \begin{cases}X_1= 2^{4m-3} + 2^{m-2}- 2^{3m-2} - 2^{2m-3} \\ 
X_2 = 3\cdot 2^{4m-2} +2^{2m-2} - 1  \\ 
X_3 = 2^{4m-3}+2^{3m-2}-2^{2m-3}-2^{m-2}. \end{cases}
	$$
	Then the desired weight distribution of $\mathcal{C}_f$ in Table \ref{Th_n=2m_table1} follows directly from (\ref{WtDist-Relation1}).

(2) For the linear code $\mathcal{C}_{D(f)}$, since $W_f(b)=2^n$ if and only if $b=0$, it follows from  (\ref{dimension}) that the dimension of $\mathcal{C}_{D(f)}$ is $n$. Note that for any $b\in\gf_{2^n}$, $W_f(b)\in\left\{ 0, 2^n, \pm 2^{\frac{n+2}{2}}  \right\}$. By (\ref{wt_new}), the weights of the codewords $\mathbf{c}_b$ in $\mathcal{C}_{D(f)}$ satisfy
	$$\mathrm{wt}(\mathbf{c}_b)\in \left\{ 2^{n-2}, 0, 2^{n-2}-2^{\frac{n-2}{2}},2^{n-2}+2^{\frac{n-2}{2}}  \right\}.$$
	Denote
	$$\mathrm{w}_1 = 2^{n-2}-2^{\frac{n-2}{2}},~~ \mathrm{w}_2=2^{n-2},~~ \mathrm{w}_3 = 2^{n-2}+2^{\frac{n-2}{2}}.$$
The desired weight distribution of $\mathcal{C}_{D(f)}$ in Table \ref{Th_n=2m_table2} can be easily obtained by solving (\ref{Eq-WtDist2}) accordingly. 
\end{proof}

{\begin{example}
		When $m=3$, the code $\mathcal{C}_f$ in Theorem \ref{Th_n=2m} is a $\left[63, 12, 24\right]$ binary linear code with the weight enumerator 
		$$ 1+ 630 z^{24} + 3087 z^{32} +  378 z^{36}. $$
		Referring to the code table \cite{Grassl:codetables}, the linear code has the best known parameter.
\end{example}}

\subsection{The case $n=2m+1$}

{From the four classes of two-to-one functions  in Lemma \ref{2-1-2}, this subsection presents five classes of $3$-weight linear codes, two classes of constant-weight linear codes and one class of at most $5$-weight linear codes.}

\begin{table}[!t]
	\centering
	\caption{The weight distribution of the codes $\mathcal{C}_{D(f)}$ in Theorem \ref{Th_n=2m+1_1}} \label{Th_n=2m+1_table1}
	\begin{tabular}{cc}
		\hline
		Weight & Multiplicity \\
		\hline
		$0$ & $1$\\
		$2^{n-1}-2^{m}$ &  $2^{4m}+2^{3m}-2^{2m-1}-2^{m-1}$ \\
		$2^{n-1}$ & $  2^{4m+1} +2^{2m} - 1  $ \\
		$2^{n-1}+2^{m}$ & $2^{4m}+2^{m-1}-2^{3m}-2^{2m-1}$\\
		\hline
	\end{tabular}
\end{table}

\begin{table}[!t]
	\centering
	\caption{The weight distribution of the codes $\mathcal{C}_{D(f)}$ in Theorem \ref{Th_n=2m+1_1}} \label{Th_n=2m+1_table2}
	\begin{tabular}{cc}
		\hline
		Weight & Multiplicity \\
		\hline
		$0$ & $1$\\
		$2^{n-2}-2^{m-1}$ & $2^{n-2}+2^{m-1}$  \\
		$2^{n-2}$ & $ 2^{n-1}-1$ \\
		$2^{n-2}+2^{m-1}$ & $2^{n-2}-2^{m-1}$\\
		\hline
	\end{tabular}
\end{table}

\begin{Th}
	\label{Th_n=2m+1_1}
	Let $n=2m+1$ and $f(x) = x^{2^{m+1}+2}+x^{2^{m+1}}+x^2+x$. 	Define two linear codes $\mathcal{C}_{f}$ and $\mathcal{C}_{D(f)}$ as in \eqref{Eq-GenCons1} and \eqref{Eq-GenCons2}, respectively. 
	Then, 
	\begin{enumerate}[(1)]
		\item  $\mathcal{C}_{f}$ is a $\left[2^{n}-1, 2n\right]$ binary linear code with weight distribution in Table \ref{Th_n=2m+1_table1}.
		\item  $\mathcal{C}_{D(f)}$ is a $\left[ 2^{n-1}-1,n\right]$ binary linear code with weight distribution in Table \ref{Th_n=2m+1_table2}.
	\end{enumerate}
\end{Th}

\begin{proof} 
In a similar manner as in Theorem \ref{Th_n=2m}, we will first investigate the value of $W_f(a,b)$ and then discuss the parameters of $\mathcal{C}_{f}$ and $\mathcal{C}_{D(f)}$.

It is clear that $W_f(a,b)=2^n$ when $(a,b) = (0,0)$. For $(a,b)\neq(0,0)$, let $\varphi_b(x) = \tr_n(bf(x))$.  Since $f$ is quadratic, according to (\ref{Eq-Quad-2}), we need to compute the dimension of the kernel of the bilinear form of $\varphi_b(x)$. Note that the bilinear form of $\varphi_b(x)$ is given by 
\begin{eqnarray*}
B_{\varphi_b}(x,y)& =& \varphi_b(x+y) + \varphi_b (x) +\varphi_b (y) \\
&=& \tr_n\left( b y^2x^{2^{m+1}} + by^{2^{m+1}}x^2 \right) = \tr_n\left(L_b(y)x^{2^{m+1}}\right),
\end{eqnarray*}
where $L_b(y)=by^2+b^{2^m}y$. Clearly, $\ker(L_b) = \left\{0, b^{2^m-1} \right\}$. According to (\ref{Eq-Quad-2}), we have
$$ W_f(a,b)  = \begin{cases}
\pm 2^{\frac{n+1}{2}}, & \text{if } \tr_n(bf(x)+ax)=0 \text{ for all } x \in \left\{0, b^{2^m-1} \right\}, \\
0, & \text{otherwise.}  
\end{cases}$$
Moreover, for $x= b^{2^m-1}$, 
\begin{eqnarray*}
\tr_n(bf(x)+ax) &=& \tr_n\left( b^{-1}  + \left(b+b^{2^m}+b^{2^{m+1}}\right)b^{1-2^{m+1}} +ab^{2^m-1} \right)\\
&=& \tr_n\left(b^{-1}+b+b^{1-2^m}+b^{2-2^{m+1}}+ab^{2^m-1}\right)\\
&=&\tr_n\left(b^{-1}+b+ab^{2^m-1}\right).
\end{eqnarray*}
Obviously, if $(a,b) = (0,b) \neq(0,0)$, $\tr_n\left(b^{-1}+b+ab^{2^m-1}\right)\in\{0,1\}$ and thus for any $(a,b)\neq(0,0)$, $$W_f(a,b), W_f(b) \in \left\{ 0, \pm 2^{\frac{n+1}{2}} \right\}. $$

Next, we consider the parameters of $\mathcal{C}_f$ and $\mathcal{C}_{D(f)}$, respectively. 

(1) For the linear code $\mathcal{C}_f,$ since $W_f(a,b)=2^n$ if and only if $(a,b)=(0,0)$, the dimension is $2n$ from (\ref{dimensionK1}). Moreover, for any $a,b\in\gf_{2^n}$, $W_f(a,b)\in\left\{ 0, 2^n, \pm 2^{\frac{n+1}{2}} \right\}$. Let 
$$v_1 = -2^{\frac{n+1}{2}},~~v_2 = 0,~~ v_3 = 2^{\frac{n+1}{2}}.$$
Then by computing (\ref{Eq-WtDist1}), we can get that the occurrences of $W_f(a,b)=v_i$'s, $i=1,2,3$ in the Walsh spectrum of $f$ are 
	$$ \begin{cases}X_1= 2^{4m}+2^{m-1}-2^{3m}-2^{2m-1} \\ 
X_2 =  2^{4m+1} +2^{2m} - 1  \\ 
X_3 = 2^{4m}+2^{3m}-2^{2m-1}-2^{m-1}. \end{cases}
$$
Finally, by (\ref{WtDist-Relation1}), the desired weight distribution of $\mathcal{C}_f$ can be obtained. 

(2) For the linear code $\mathcal{C}_{D(f)}$, $W_f(b)= 2^n$ if and only if $b=0$, which means that the dimension of $\mathcal{C}_{D(f)}$ is $n$ according to (\ref{dimension}). Since for any $b\in\gf_{2^n}$, $W_f(b)\in\left\{ 0, 2^n, \pm 2^{\frac{n+1}{2}}  \right\}$, by (\ref{wt_new}), the weights of the codewords $\mathbf{c}_b$ in $\mathcal{C}_{D(f)}$ satisfy
$$\mathrm{wt}(\mathbf{c}_b)\in \left\{ 2^{n-2}, 0, 2^{n-2}-2^{\frac{n-3}{2}},2^{n-2}+2^{\frac{n-3}{2}}  \right\}.$$

In the following, we determine the weight distribution of $\mathcal{C}_{D(f)}$.  Define 
$$\mathrm{w}_1 = 2^{n-2}-2^{\frac{n-3}{2}},~~ \mathrm{w}_2=2^{n-2},~~ \mathrm{w}_3 = 2^{n-2}+2^{\frac{n-3}{2}}.$$
Then solving (\ref{Eq-WtDist2})  gives the desired weight distribution. 
\end{proof}

{\begin{example}
		When $m=3$, the code $\mathcal{C}_f$ in Theorem \ref{Th_n=2m+1_1} is a $\left[127, 14, 56\right]$ binary linear code with the weight enumerator 
	$$ 1+ 4572 z^{56} + 8255 z^{64} +  3556 z^{72}. $$
	Referring to the code table \cite{Grassl:codetables}, the linear code is optimal.
\end{example}}

\begin{Th}
	\label{Th_n=2m+1_2}
	Let $n=2m+1$ and $f(x) = x^{2^n-2^{m+1}+2}+x^{2^{m+1}}+x^2+x$. 	Define two linear codes $\mathcal{C}_{f}$ and $\mathcal{C}_{D(f)}$ as in \eqref{Eq-GenCons1} and \eqref{Eq-GenCons2}, respectively. 
	Then, 
	\begin{enumerate}[(1)]
		\item { $\mathcal{C}_{f}$ is a $\left[2^{n}-1, 2n\right]$ binary linear code with at most five weights.}
		\item  $\mathcal{C}_{D(f)}$ is a $\left[ 2^{n-1}-1,n\right]$ binary linear code with weight distribution in Table \ref{Th_n=2m+1_table2}.
	\end{enumerate}
\end{Th}

\begin{proof}
	 Since $\gcd\left( 2^n-2^{m+1}+2, 2^n-1 \right) = 1$ and $$\left( 2^n-2^{m+1}+2 \right)\times \left(2^m+1\right) \equiv 2^m+2 \pmod {2^n-1},$$ we have 
	\begin{eqnarray*}
		W_f(a,b)&=&\sum_{x\in \gf_{2^n}} (-1)^{\tr_n\left(a x + b\left(x^{2^n-2^{m+1}+2}+x^{2^{m+1}}+x^2+x \right)\right)} \\
		&=& \sum_{x\in \gf_{2^n}}(-1)^{\tr_n\left( a x^{2^m+1} + b \left(  x^{2^m+2} + x^{2^{m+1}+1} + x^{2^{m+1}+2}+ x^{2^m+1}\right)  \right) }\\
		&=&\sum_{x\in \gf_{2^n}}(-1)^{\tr_n\left( (a^2+ b+b^2) x^{2^{m+1}+2} + bx^{2^{m+1}+1} + bx^{2^m+2}   \right)}.
	\end{eqnarray*}
Define $$Q(x) = (a^2+ b+b^2) x^{2^{m+1}+2} + bx^{2^{m+1}+1} + bx^{2^m+2}.$$ It is clear that when $(a,b) = (0,0)$, $W_f(a,b)=0$. For $(a,b)\neq(0,0)$, the bilinear form of $\varphi_{a,b}(x)=\tr_n(Q(x))$ is given by 
\begin{eqnarray*}
B_{\varphi_{a,b}}(x,y) &=& \varphi_{a,b}(x+y)+\varphi_{a,b}(x)+\varphi_{a,b}(y) \\
&=& \tr_n\left( (a^2+b+b^2) \left( y^2x^{2^{m+1}} +y^{2^{m+1}} x^2 \right) + b\left(yx^{2^{m+1}}+y^{2^{m+1}}x\right) + b\left( y^2x^{2^m} + y^{2^m}x^2 \right)  \right) \\
&=& \tr_n\left( L_{a,b}(y)x^{2^{m+2}} \right),
\end{eqnarray*}
where 
$$L_{a,b}(y) = b^4y^8+\left(b^{2^{m+2}}+b^4+b^2+a^4\right)y^4+\left( b^{2^{m+2}} +b^{2^{m+1}}+b^2 + a^{2^{m+2}}\right)y^2+ b^{2^{m+1}}y.$$
Let $\ker(L_{a,b}) = \left\{ y\in\gf_{2^n}: L_{a,b}(y)=0 \right\}$. According to (\ref{Eq-Quad}), we have 
$$ W_f(a,b)  = 2^n\sum_{y\in \ker(L_{a,b})} (-1)^{\varphi_{a,b}(y)}
=  \begin{cases}
\pm 2^{\frac{n+d_{a,b}}{2}}, & \text{if } \varphi_{a,b}(x)=0 \text{ for all } x \in \ker(L_{a,b}), \\
0, & \text{otherwise,} 
\end{cases}$$
where $d_{a,b}$ is the dimension of $\ker(L_{a,b})$. From the expression of $L_{a,b}$, it is obvious that $d_{a,b}\le3$. Moreover, since $n+d$ must be even and $n$ is odd, $d_{a,b}\in\{1,3\}$. 
Hence the Walsh transform
$$
W_f(a,b) \in \left\{2^n, \,0, \,\pm 2^{\frac{n+1}{2}}, \,\pm 2^{\frac{n+3}{2}} \right\}.
$$

Next, we will show that when $a=0$ the Walsh transform $W_f(0,b)=0$ when the dimension of $L_{a,b}(y)=0$ equals 3. 
Namely, there exists some $y_0\in\ker(L_{a,b})$ such that $\varphi_{a,b}(y_0)=1$. In this case, we have
$$
L_{b}(y)=L_{0,b}(y)= b^4y^8+\left(b^{2^{m+2}}+b^4+b^2\right)y^4+\left( b^{2^{m+2}} +b^{2^{m+1}}+b^2 \right)y^2+ b^{2^{m+1}}y.
$$ 
Denote
$\varphi_b(y)=\varphi_{0,b}(y)=\tr_n\left( b \left(y^{2^{m+1}+2} + y^{2^{m+1}+1} + y^{2^m+2} +y^{2^m+1} \right)\right)$. Let $z=y^2+y$. Then we have 
	$$L_b= b^4z^4 + b^2z^2 + b^{2^{m+2}}z^2 + b^{2^{m+1}}z = (b^2z^2+b^{2^{m+1}}z)^2+b^2z^2+b^{2^{m+1}}z.$$
	If $d_{0,b}=3$, i.e., the number of solutions of $L_b=0$ equals $8$, then the equation 
	\begin{equation}
		\label{bz}
		(b^2z^2+b^{2^{m+1}}z)^2+b^2z^2+b^{2^{m+1}}z=0
	\end{equation} has $4$ solutions in $\gf_{2^n}$ since $y$ and $y+1$ correspond to one same $z=y^2+y$. Clearly, from (\ref{bz}), we have $b^2z^2+b^{2^{m+1}}z=0$ or $b^2z^2+b^{2^{m+1}}z=1$. From $b^2z^2+b^{2^{m+1}}z=0$, we get two solutions $z_0=0$ and $z_1= b^{2^{m+1}-2}$ in $\gf_{2^n}$. Similarly, we also obtain two solutions $z=z_2, z_3$ from $b^2z^2+b^{2^{m+1}}z=1$. Thus if $d_b=3$, $y^2+y=z_i$ for $i=0,1,2,3$ exactly has two solutions in $\gf_{2^n}$. Namely, $\tr_n(z_i)=0$ for $i=0,1,2,3$. Particularly, $\tr_n(z_1)=\tr_n(b^{2^{m+1}-2})=0$. Therefore, there exists some element $y_0\in\gf_{2^n}$ such that $b^{2^{m+1}-2} = y_0^2+y_0$, i.e., $b= \frac{1}{(y_0^2+y_0)^{2^m+1}}$. In fact, such $y_0$ belongs to $\ker(L_b)$ and is what we need. Indeed, 
	\begin{eqnarray*}
		\varphi_b(y_0) &=& \tr_n\left( b \left(y_0^{2^{m+1}+2} + y_0^{2^{m+1}+1} + y_0^{2^m+2} +y_0^{2^m+1}\right) \right) \\
		&=& \tr_n\left( \frac{y_0^{2^{m+1}+2} + y_0^{2^{m+1}+1} + y_0^{2^m+2} +y_0^{2^m+1} }{(y_0^2+y_0)^{2^m+1}}\right)\\
		&=& \tr_n(1) =1. 
	\end{eqnarray*}
	Hence, $W_f(b)=0$ when the dimension of $\ker(L_b)$ is $3$.  

Next, we consider the parameters of $\mathcal{C}_f$ and $\mathcal{C}_{D(f)}$, respectively. 

(1) For the linear code $\mathcal{C}_f,$ since $W_f(a,b)=2^n$ if and only if $(a,b)=(0,0)$, the dimension is $2n$ from (\ref{dimensionK1}). Moreover, 
the possible Hamming weights of codewords in $\mathcal{C}_f$ are given by
$$
{\rm wt}(\mathbf{c}_{a,b}) \in \left\{0, 2^{n-1}, 2^{n-1}\pm 2^{\frac{n-1}{2}}, 2^{n-1}\pm 2^{\frac{n+1}{2}} \right\}.
$$

(2) For the linear code $\mathcal{C}_{D(f)}$, $W_f(b)= 2^n$ if and only if $b=0$, which means that the dimension of $\mathcal{C}_{D(f)}$ is $n$ according to (\ref{dimension}). 
Since for any $b\in\gf_{2^n}$, $W_f(b)\in\left\{ 2^n, 0, \pm 2^{\frac{n+1}{2}}  \right\}$, by (\ref{wt_new}), the weights of the codewords $\mathbf{c}_b$ in $\mathcal{C}_{D(f)}$ satisfy
$$\mathrm{wt}(\mathbf{c}_b)\in \left\{ 2^{n-2}, 0, 2^{n-2}-2^{\frac{n-3}{2}},2^{n-2}+2^{\frac{n-3}{2}}  \right\}.$$
Similar to the proof in Theorem \ref{Th_n=2m+1_1} (2), the desired weight distribution of $\mathcal{C}_{D(f)}$ can be obtained by solving \eqref{Eq-WtDist2} accordingly.	
\end{proof}

\begin{table}[!t]
	\centering
	\caption{The weight distribution of the codes $\mathcal{C}_{f}$ in Theorem \ref{Th_n=2m+1_3}} \label{Th_n=2m+1_table3}
	\begin{tabular}{cc}
		\hline
		Weight & Multiplicity \\
		\hline
		$0$ & $1$\\
		$2^{n-1}-2^{m}$ &  $2^{4m-1}+2^{3m-1}-2^{2m-1}-2^{m-1}$ \\
		$2^{n-1}$ & $  2^{4m} +2^{2m} - 1  $ \\
		$2^{n-1}+2^{m}$ & $2^{4m-1}+2^{m-1}-2^{3m-1}-2^{2m-1}$\\
		\hline
	\end{tabular}
\end{table}

\begin{Th}
	\label{Th_n=2m+1_3}
	Let $n=2m+1$ and $f(x) = x^{2^{m+1}+2}+x^{2^{m+1}+1}+x^2+x$. 	Define two linear codes $\mathcal{C}_{f}$ and $\mathcal{C}_{D(f)}$ as in \eqref{Eq-GenCons1} and \eqref{Eq-GenCons2}, respectively. 
	Then, 
	\begin{enumerate}[(1)]
		\item  $\mathcal{C}_{f}$ is a $\left[2^{n}-1, 2n-1\right]$ binary linear code with weight distribution in Table \ref{Th_n=2m+1_table3}.
		\item  $\mathcal{C}_{D(f)}$ is a $\left[ 2^{n-1}-1,n-1\right]$ binary linear code with  weight enumerator $1+(2^{n-1}-1)z^{2^{n-2}}$.
	\end{enumerate}
\end{Th}

\begin{proof}
	Firstly, we shall compute the value $W_f(a,b)$. It is clear that when $(a,b)=(0,0)$, $W_f(a,b)=2^n$. For $(a,b)\neq(0,0)$, let $\varphi_b(x)=\tr_n(bf(x))$. Since $f$ is quadratic, according to (\ref{Eq-Quad-2}), we need to determine the dimension of kernel of the bilinear form of $\varphi_b(x)$. Note that the bilinear form of $\varphi_b(x)$ is given by 
	\begin{eqnarray*}
	B_{\varphi_b}(x,y) &=& \varphi_b(x+y)  + \varphi_b(x) + \varphi_b (y) \\
	&=& \tr_n\left( b\left( y^2x^{2^{m+1}} +  y^{2^{m+1}}x^2 \right) + b\left( yx^{2^{m+1}} + y^{2^{m+1}}x \right)  \right) = \tr_n\left(L_b(y)x^{2^{m+1}}\right),
	\end{eqnarray*}
where $$L_b(y)=\left( b^{2^{m+1}} + b \right)y^2 + \left( b^{2^m} + b \right) y.$$   Since $\gcd(m,n)=\gcd(m+1,n)=1$, $b^{2^{m+1}}+b=0$ and $b^{2^m}+b=0$ both have two solutions $b=0,1$ in $\gf_{2^n}$. Thus $L_b(y)=0$ if and only if $b = 0,1$. Moreover, according to the expression of $W_f(a,b)$, we have 
$$W_f(a,1) = \sum_{x\in \gf_{2^n}}(-1)^{\tr_n\left( ax + x^{2^{m+1}+2} + x^{2^{m+1}+2} \right)} = \sum_{x\in\gf_{2^n}}(-1)^{\tr_n(ax)} $$
equals $2^n$ if $a=0$ and $0$ otherwise. 

When $b\neq 0,1$, then it is clear that $\ker(L_b) = \left\{ 0, \frac{b^{2^m}+b}{b^{2^{m+1}}+b} \right\}$. It follows from (\ref{Eq-Quad-2}) that 
$$ W_f(a,b)  = \begin{cases}
\pm 2^{\frac{n+1}{2}}, & \text{if } \tr_n(bf(x)+ax)=0 \text{ for all } x \in \left\{0, \frac{b^{2^m}+b}{b^{2^{m+1}}+b} \right\}, \\
0, & \text{otherwise.}  
\end{cases}$$
Let  $y_0 = \frac{b^{2^m}+b}{b^{2^{m+1}}+b}= \left( b^{2^{m+1}} +b \right)^{2^m-1}$. Note that $$y_0^{2^{m+1}+2}=\frac{1}{b^{2^{m+1}}+b}, ~~ y_0^{2^{m+1}+1} = \frac{1}{b^{2^{m}}+b} $$ and 
\begin{eqnarray*}
	&&  \tr_n\left( by_0^{2^{m+1}+2} + by_0^{2^{m+1}+1} \right) \\
	&=& \tr_n\left( \frac{b}{b^{2^{m+1}}+b} + \frac{b}{b^{2^{m}}+b}  \right)\\
	&=& \tr_n\left( \frac{b^{2^m}+b}{b^{2^{m}}+b} \right) = 1.
\end{eqnarray*}
Moreover, denote $b^{2^{m+1}}+b= \Delta$. Then $y_0=\Delta^{2^m-1}$, $b^{2^{m+1}}+b^2=\Delta^{2^{m+1}}$ and $b^2+b=\Delta+\Delta^{2^{m+1}}$. Thus 
$$\tr_n\left( \left( b+b^2 \right)y_0^2\right) = \tr_n\left(\frac{\left( \Delta+\Delta^{2^{m+1}}\right)\Delta^{2^{m+1}}}{\Delta^2}\right) = \tr_n\left(\Delta^{2^{m+1}-1}+\Delta^{2^{m+2}-2}\right)=0.$$
Therefore, 
$$ \tr_n(bf(y_0)) = \tr_n\left( by_0^{2^{m+1}+2} + by_0^{2^{m+1}+1} + \left( b+b^2 \right)y_0^2 \right) = 1+0=1$$
and 
$$\tr_n(bf(y_0)+ay_0) = 1+ \tr_n(ay_0).$$
Clearly, $\tr_n(bf(y_0)+ay_0) \in \{0,1\}$ and thus for any $b\neq 0,1$, $$W_f(a,b)\in\left\{ 0, \pm 2^{\frac{n+1}{2}} \right\}.$$
Moreover, since $\tr_n(bf(y_0))=1$ here, $W_f(b)=0$ for $b\in\gf_{2^n}\backslash\{1\}$.

(1) For the linear code $\mathcal{C}_f$, since $W_f(a,b) = 2^n$ if and only if $(a,b)=(0,1), (0,0)$, the dimension is $2n-1$ from (\ref{dimensionK1}). Moreover, for any $a,b\in\gf_{2^n}$, $W_f(a,b)\in\left\{ 0,2^n, \pm 2^{\frac{n+1}{2}} \right\}$. Let 
$$v_1 = -2^{\frac{n+1}{2}},~~v_2 = 0,~~ v_3 = 2^{\frac{n+1}{2}}.$$
Then by computing (\ref{Eq-WtDist1}), we can get that the occurrences of $W_f(a,b)=v_i$'s, $i=1,2,3$ in the Walsh spectrum of $f$ are 
$$ \begin{cases}X_1= 2^{4m}+2^{m}-2^{3m}-2^{2m} \\ 
X_2 =  2^{4m+1} +2^{2m+1} - 2  \\ 
X_3 = 2^{4m}+2^{3m}-2^{2m}-2^{m}. \end{cases}
$$
Finally, by (\ref{WtDist-Relation1}), the desired weight distribution of $\mathcal{C}_f$ can be obtained. 

(2) For the linear code $\mathcal{C}_{D(f)}$, since there are two $b$'s ($0$ and $1$) such that $W_f(b)=2^n$ and $(2^n-2)$ $b$'s such that $W_f(b)=0$, the dimension of $\mathcal{C}_{D(f)}$ equals $n-1$ by (\ref{dimension}). Moreover, by (\ref{wt_new}), the Hamming weights of the codewords $\mathbf{c}_b$ in $\mathcal{C}_{D(f)}$ satisfy $$\mathrm{wt}(\mathbf{c}_b)\in \{0, 2^{n-2}\}.$$ The desired weight enumerator thus follows.
\end{proof}

\begin{Th}
	\label{Th_n=2m+1_4}
	Let $n=2m+1$ and $f(x) = x^{2^{m+2}+4}+x^{2^{m+1}+2}+x^2+x$. 	Define two linear codes $\mathcal{C}_{f}$ and $\mathcal{C}_{D(f)}$ as in \eqref{Eq-GenCons1} and \eqref{Eq-GenCons2}, respectively. 
	Then, 
	\begin{enumerate}[(1)]
		\item  $\mathcal{C}_{f}$ is a $\left[2^{n}-1, 2n-1\right]$ binary linear code with weight distribution in Table \ref{Th_n=2m+1_table3}.
		\item  $\mathcal{C}_{D(f)}$ is a $\left[ 2^{n-1}-1,n-1\right]$ binary linear code with  weight enumerator $1+(2^{n-1}-1)z^{2^{n-2}}$.
	\end{enumerate}
\end{Th}

\begin{proof}
	The proof is very similar with that of Theorem \ref{Th_n=2m+1_3} and we omit it here.
\end{proof}

{\begin{example}
		When $m=3$, the codes $\mathcal{C}_f$ in Theorems \ref{Th_n=2m+1_3} and \ref{Th_n=2m+1_4} are two $\left[127, 13, 56\right]$ binary linear codes with the same weight enumerator 
		$$ 1+ 2268 z^{56} + 4159 z^{64} +  1764 z^{72}. $$
		Referring to the code table \cite{Grassl:codetables}, the linear codes $\mathcal{C}_f$ have the best known parameter.
		When $m=3$, the codes $\mathcal{C}_{D(f)}$ in Theorems \ref{Th_n=2m+1_3} and \ref{Th_n=2m+1_4} are two $\left[63,6,32\right]$ binary linear codes with the same weight enumerator $$1+63z^{32}.$$
		Referring to the code table \cite{Grassl:codetables}, the linear codes $\mathcal{C}_{D(f)}$ are optimal.
\end{example}}

\subsection{The case $n=3m$}

In this subsection, we consider binary linear codes from the first two-to-one polynomial in Lemma \ref{2-1-3}.
The second one will be generalized in Section 4 and the corresponding linear code will be discussed later.

\begin{Th}\label{Th_n=3m}
	Let $n=3m$ with $m\equiv0\pmod3$ and $f(x)=x^{2^{2m}+1}+x^{2^{m+1}}+x^{2^m+1}+x$. Define two linear codes $\mathcal{C}_{f}$ and $\mathcal{C}_{D(f)}$ as in \eqref{Eq-GenCons1} and \eqref{Eq-GenCons2}, respectively. 
	Then, 
	\begin{enumerate}[(1)]
		\item  $\mathcal{C}_{f}$ is a $\left[2^{n}-1, 5m\right]$ binary linear code with weight distribution in Table \ref{Th_n=3m_table1}.
		\item  $\mathcal{C}_{D(f)}$ is a $\left[ 2^{n-1}-1,n-1\right]$ binary linear code with  weight enumerator $1+(2^{n-1}-1)z^{2^{n-2}}$.
	\end{enumerate}
\end{Th}

\begin{table}[!t]
	\centering
	\caption{The weight distribution of the codes $\mathcal{C}_{f}$ in Theorem \ref{Th_n=3m}} \label{Th_n=3m_table1}
	\begin{tabular}{cc}
		\hline
		Weight & Multiplicity \\
		\hline
		$0$ & $1$\\
		$2^{n-1}-2^{2m-1}$ &  $2^{4m-1}+2^{3m-1}-2^{2m-1}-2^{m-1}$ \\
		$2^{n-1}$ & $  2^{5m} +2^{2m} - 2^{4m} - 1  $ \\
		$2^{n-1}+2^{2m-1}$ & $2^{4m-1}+2^{m-1}-2^{3m-1}-2^{2m-1}$\\
		\hline
	\end{tabular}
\end{table}

\begin{proof}
	Firstly, we shall compute the value $W_f(a,b)$. Clearly, when $(a,b)=(0,0)$, $W_f(a,b)=2^n$. For $(a,b)\neq(0,0)$, let $\varphi_b(x)=\tr_n(bf(x))$. If $(a,b)=(a,1)$, then $$W_f(a,b) = \sum_{x\in \gf_{2^n}}(-1)^{\tr_n\left( ax + x^{2^{2m}+1}+x^{2^{m+1}}+x^{2^m+1}+x  \right)} = \sum_{x\in \gf_{2^n}}(-1)^{\tr_n(ax)}, $$
	which equals $2^n$ if $a=0$ and $0$ otherwise. Moreover, if $b\in\gf_{2^m}\backslash\gf_{2}$, then 
	$$W_f(a,b)=\sum_{x\in \gf_{2^n}}(-1)^{\tr_n\left(\left(a^{2^{m+1}}+b+b^2\right)x^{2^{m+1}}\right)},$$
	which equals $2^n$ if $a^{2^{m+1}}+b+b^2=0$ and $0$ otherwise. Next, we assume $b\in\gf_{2^n}\backslash\gf_{2^m}$.
	
	 Since $\tr_n(f)$ is a quadratic function, according to (\ref{Eq-Quad-2}), we need to determine the dimension of kernel of the bilinear form of $\varphi_b(x)$. Note that the bilinear form of $\varphi_b(x)$ is given by 
	\begin{eqnarray*}
	B_{\varphi_b}(x,y) &=& \varphi_b(x+y) + \varphi_b(x) + \varphi_b(y) \\
	&=& \tr_n\left( \left( b+b^{2^m} \right)\left(yx^{2^m} + y^{2^m}x \right) \right)   = \tr_n\left( L_b(y)x^{2^{m}}\right),
	\end{eqnarray*}
	where $L_b(y)= \left(b+b^{2^m}\right)y + \left( b^{2^m} + b^{2^{2m}} \right) y^{2^{2m}}$. Let $\ker(L_b) = \left\{ y: y \in \gf_{2^n} ~\text{and}~ L_b(y) = 0 \right\}$. Since $b\in\gf_{2^n}\backslash\gf_{2^m}$, we have 
	$$\ker(L_b) = \left\{ \left(b^{2^m}+b^{2^{2m}}\right) \eta: \eta\in\gf_{2^m}  \right\}.$$
 It follows from (\ref{Eq-Quad-2}) that 
$$ W_f(a,b)  = \begin{cases}
\pm 2^{2m}, & \text{if } \tr_n(bf(x)+ax)=0 \text{ for all } x \in \ker(L_b), \\
0, & \text{otherwise.}  
\end{cases}$$
For any $x=\left(b^{2^m}+b^{2^{2m}}\right) \eta\in\ker(L_b)$ with $\eta\in\gf_{2^m}$, 
$$\tr_n(bf(x)+ax) = \tr_n\left( U_{a,b} \eta^2 \right) = \tr_m\left( \tr_{n/m}(U_{a,b}) \eta^2 \right).$$
where 
$$U_{a,b} =  a^2\left(b+b^{2^m}\right)^{2^{m+1}} +  \left( b+b^{2^m} \right)^{2^{2m}+2^m+1} + \left( b+b^{2^m} \right)^{2^{2m+1}}\left(b^{2^{m+1}}+b\right). $$
Clearly, $\tr_n(bf(x)+ax)=0$ if and only if $a^2= \left(b+b^{2^m}\right)^{2^{2m}-2^m+1} + \left(b+b^{2^m}\right)^{2^{2m+1}-2^{m+1}}\left(b^{2^{m+1}}+b\right)$. Thus for any $a\in\gf_{2^n}, b\in\gf_{2^n}\backslash\gf_{2^m}$, $$W_f(a,b)\in\left\{ 0,\pm 2^{2m} \right\}.$$
As for $W_f(b)$, we need the following claim which will be showed at the end of the proof. 

	{\bfseries Claim.} For any $b\in\gf_{2^n}\backslash\gf_{2^m}$, $\tr_{n/m}(U_b)\neq0$, where $U_b=U_{0,b}$.
	
		According to the above claim, it is clear that $W_f(b)=0$  for any $b\in\gf_{2^n}\backslash\gf_{2^m}$.
		
		Next, we consider the parameters of $\mathcal{C}_f$ and $\mathcal{C}_{D(f)}$, respectively. 
		
		(1) For the linear code $\mathcal{C}_f$, since $W_f(a,b)=2^n$ if and only if  $b\in\gf_{2^m}$ and $a^{2^{m+1}}+b+b^2=0$, by (\ref{dimensionK1}), the dimension of $K_1 = \left\{ a,b\in\gf_{2^n}: W_f(a,b)=2^n \right\}$ is $m$ and thus the dimension of $\mathcal{C}_f$ is $2n-m=5m$. Moreover, for any $a,b\in\gf_{2^n}$, $W_f(a,b)\in\left\{ 0, 2^n, \pm 2^{2m} \right\}$. Let 
		$$v_1=-2^{2m}, ~~~ v_2 = 0, ~~~ v_3 = 2^{2m}. $$
Then by computing (\ref{Eq-WtDist1}), we can obtain that the occurrences of $W_f(a,b)=v_i$'s, $i=1,2,3$ in the Walsh spectrum of $f$ are 	$$ \begin{cases}X_1= 2^{5m-1}+2^{2m-1}-2^{4m-1}-2^{3m-1} \\ 
X_2 = 2^{6m}+2^{3m}-2^{5m}-2^m  \\ 
X_3 = 2^{5m-1}+2^{4m-1}-2^{3m-1}-2^{2m-1}. \end{cases}
$$
Finally, by (\ref{WtDist-Relation1}), the desired weight distribution of $\mathcal{C}_f$ can be obtained. 
		
		(2) For the linear code $\mathcal{C}_{D(f)}$, since there are two $b$'s ($0$ and $1$) such that $W_f(b)=2^n$ and ($2^n-2$) $b$'s such that $W_f(b)=0$, by (\ref{dimension}), the dimension of $\mathcal{C}_{D(f)}$ equals $n-1$. Moreover, by  (\ref{wt_new}), we know that the weights of the codewords $\mathbf{c}_b$ in $\mathcal{C}_{D(f)}$ satisfy $$\mathrm{wt}(\mathbf{c}_b)\in \{0, 2^{n-2}\}.$$ Furthermore, the weight enumerator is clear.  
	
		Finally, we prove the claim, i.e., for any $b\in\gf_{2^n}\backslash\gf_{2^m}$, $\tr_{n/m}(U_b)\neq0$. By direct computation, we have 
	\begin{equation}
	\label{Ub}
	\tr_{n/m}(U_b) = b^3+b^{3\cdot 2^m} + b^{3\cdot 2^{2m}} + b^{2^{2m+1}+2^m}+ b^{2^{2m}+2} + b^{2^{m+1}+1}.
	\end{equation}
	For any $b\in\gf_{2^n}\backslash\gf_{2^m}$, define 
	\begin{equation*}
	\begin{cases} 
	b + b^{2^m} + b^{2^{2m}} = \alpha \\ 
	b^{2^m+1}+b^{2^{2m}+1}+b^{2^{2m}+2^m} = \beta  \\ 
	b^{2^{2m}+2^m+1} = \gamma \end{cases}
	\end{equation*}
	and $g(x)=x^3+\alpha x^2+\beta x+\gamma\in\gf_{2^m}[x]$.  Then it is clear that $g(x)=(x+b)\left(x+b^{2^m}\right)\left(x+b^{2^{2m}}\right)$ and $g(x)$ is irreducible on $\gf_{2^m}$. Let $ u = b^{2^{2m+1}+2^m}+ b^{2^{2m}+2} + b^{2^{m+1}+1}$ and $v= b^{2^{2m+1}+1}+ b^{2^{m}+2} + b^{2^{m+1}+2^{2m}}$. Then  we have 
	\begin{eqnarray*}
		u+v &=& \left(b+b^{2^m}\right)\left( b+b^{2^{2m}} \right)\left(b^{2^{m}}+b^{2^{2m}}\right)\\
		&=& (\alpha+b)\left(\alpha+b^{2^m}\right)\left(\alpha+b^{2^{2m}}\right) = g(\alpha) = \alpha\beta + \gamma.
	\end{eqnarray*}
	In addition, from the expanded form of $\left( b + b^{2^m} + b^{2^{2m}} \right)^3$, we know that 
	$$b^3+b^{3\cdot 2^m} + b^{3\cdot 2^{2m}} = \alpha^3 + u+v = \alpha^3 + \alpha\beta + \gamma.$$
	Moreover, since $b,b^{2^m}, b^{2^{2m}}$ are the roots of $g(x)$ in $\gf_{2^n}$, $\frac{1}{b},\frac{1}{b^{2^m}}, \frac{1}{b^{2^{2m}}}$ are the roots of $$g^{'}(x)=\frac{1}{\gamma}x^3g\left(\frac{1}{x}\right) = x^3+\frac{\beta}{\gamma}x^2+\frac{\alpha}{\gamma}x+\frac{1}{\gamma}.$$ in $\gf_{2^n}$. Similarly, we have 
	$$\frac{1}{b^3}+\frac{1}{b^{3\cdot 2^m}}+ \frac{1}{b^{3\cdot 2^{2m}}} =\frac{\beta^3}{\gamma^3} + \frac{\alpha\beta}{\gamma^2}+\frac{1}{\gamma} = \frac{\beta^3+\alpha\beta\gamma+\gamma^2}{\gamma^3}. $$
	Furthermore, 
	\begin{eqnarray*}
		uv &=& \left(  b^{2^{2m+1}+2^m}+ b^{2^{2m}+2} + b^{2^{m+1}+1} \right) \left(  b^{2^{2m+1}+1}+ b^{2^{m}+2} + b^{2^{m+1}+2^{2m}}\right) \\
		&=& \gamma^2+\gamma \left(b^3+b^{3\cdot 2^m} + b^{3\cdot 2^{2m}} \right) + b^{3\cdot\left(2^m+1\right)} +b^{3\cdot\left(2^{2m}+2^m\right)} + b^{3\cdot\left(2^{2m}+1\right)}\\
		&=& \gamma^2+\gamma \left(\alpha^3 + \alpha\beta + \gamma\right) + \gamma^3\left( \frac{1}{b^3}+\frac{1}{b^{3\cdot 2^m}}+ \frac{1}{b^{3\cdot 2^{2m}}} \right)\\
		&=& \alpha^3\gamma + \alpha\beta\gamma + \beta^3+\alpha\beta\gamma+\gamma^2\\
		&=&\alpha^3\gamma +  \beta^3+\gamma^2.
	\end{eqnarray*}
	Now we go back to the expression of $\tr_{n/m}(U_b)$, i.e., (\ref{Ub}). If $\tr_{n/m}(U_b)=0$, we have 
	$$u=b^3+b^{3\cdot 2^m} + b^{3\cdot 2^{2m}} = \alpha^3 + \alpha\beta + \gamma$$
	and then 
	$$v= \alpha\beta + \gamma+u = \alpha^3. $$
	Thus $uv=\alpha^6+\alpha^4\beta+ \alpha^3\gamma = \alpha^3\gamma +  \beta^3+\gamma^2$, namely,
	\begin{equation}
	\label{abg}
	\alpha^6+\alpha^4\beta + \beta^3+\gamma^2=0.
	\end{equation}
	Next, we show that under  (\ref{abg}), $g(x)=0$ has three solutions in $\gf_{2^m}$, which is in contradiction with the irreduciblity of $g(x)$. Firstly, using $x+\alpha$ to replace $x$ in $g(x)=0$ and simplifying it, we obtain 
	\begin{equation}
	\label{g=0}
	x^3+(\alpha^2+\beta)x+\alpha\beta + \gamma = 0.
	\end{equation} 
	Moreover, $$\tr_m\left( \frac{(\alpha^2+\beta)^3}{(\alpha\beta + \gamma)^2}  \right) = \tr_m\left(\frac{\alpha^6+\beta^3+\alpha^4\beta+\alpha^2\beta^2}{(\alpha\beta + \gamma)^2}\right) = \tr_m(1),$$
	where the last equality is derived from (\ref{abg}). Furthermore, it is easy to get that the equation $$t^2+(\alpha\beta+\gamma)t+(\alpha^2+\beta)^3=0$$ has a solution $t_1=(\alpha\beta+\gamma)\omega$,  where $\omega^3=1$. Since $m\equiv0\pmod3$, $\omega$ is a cube in $\gf_{2^m}$ ($m$ even), $\gf_{2^{2m}}$ ($m$ odd). In addition, $\alpha\beta+\gamma = \sqrt{(\alpha^2+\beta)^3}$ is also a cube in $\gf_{2^m}$. Thus according to Lemma \ref{cubic}, (\ref{g=0}) has three solutions in $\gf_{2^m}$, which is a contradiction, and thus $\tr_{n/m}(U_b)\neq0$. 
\end{proof}

\section{Binary linear codes from new two-to-one polynomials}
In this section, we construct two new classes of two-to-one functions , of which the first one is a generalization of  (2) in Lemma \ref{2-1-3}. Then we also obtain some binary linear codes $\mathcal{C}_{f}$ and $\mathcal{C}_{D(f)}$ from these two new two-to-one functions .  

\subsection{Two new classes of two-to-one functions }

\begin{Th}
	Let $n=km$ with $k,m$ odd and $f(x) = \tr_{n/m}\left(x^{2^m+1}\right)+x$. Then $f(x)$ is two-to-one over $\gf_{2^n}$.
\end{Th}
\begin{proof}
	 According to the definition of two-to-one functions , it suffices to prove that for any $a\in\gf_{2^n}$, $|f^{-1}(a)|\in\{0,2\}$. Namely, for any $a\in\gf_{2^n}$, $f(x+a)+f(a)=0$ has exactly two solutions in $\gf_{2^n}$. By simplifying the equation, we have 
	\begin{equation}
	\label{2-1}
	\tr_{n/m}\left( x^{2^m+1}+ax^{2^m}+a^{2^m}x \right) = x.
	\end{equation}
	This implies $x\in\gf_{2^m}$, and then  (\ref{2-1}) becomes $x^2+x=0$, which has exactly two solutions $x=0,1$.
\end{proof}

\begin{Th}
		Let $n=3m$ with $m$ odd and $f(x)=x^{2^{2m+1}+1}+x^{2^{m+1}+1}+x^4+x^3$. Then $f(x)$ is two-to-one over $\gf_{2^n}$.
\end{Th}

\begin{proof}
	 It suffices to prove that for any $a\in\gf_{2^n}$, the equation $f(x+a)+f(a)=0$, 
	\begin{equation}
	\label{eq1-1}
	x^{2^{2m+1}+1} + ax^{2^{2m+1}} + x^{2^{m+1}+1} + ax^{2^{m+1}} + x^4+x^3+ax^2+ \left(a^{2^{2m+1}}+a^{2^{m+1}} + a^2 \right)x =0, 
	\end{equation}
	has exactly two solutions in $\gf_{2^n}$. In fact, since $x=0$ is clearly a solution of (\ref{eq1-1}), we shall only show that (\ref{eq1-1}) has at most two solutions in $\gf_{2^n}$.
	
	Let $y=x^{2^m}$, $z=y^{2^m}$, $b=a^{2^m}$ and $c=b^{2^m}$. Then (\ref{eq1-1}) becomes 
	\begin{equation}
	\label{eq1-2}
	xz^2+xy^2+x^3+x^4+a(x^2+y^2+z^2)+(a^2+b^2+c^2)x =0.
	\end{equation}
	Raising (\ref{eq1-2}) to the $2^m$-th power and the $2^{2m}$-th power, we get 
	\begin{equation}
	\label{eq1-3}
	yx^2+yz^2+y^3+y^4+b(x^2+y^2+z^2)+(a^2+b^2+c^2)y =0
	\end{equation}
	and 
	\begin{equation}
	\label{eq1-4}
	zy^2+zx^2+z^3+z^4+c(x^2+y^2+z^2)+(a^2+b^2+c^2)z =0,
	\end{equation}
	respectively. Let $t=x+y+z$ and $s=a+b+c$. Computing the summation of (\ref{eq1-2}), (\ref{eq1-3}) and (\ref{eq1-4}), we obtain 
	\begin{equation*}
	t^4+t^3+st^2+s^2t =0.
	\end{equation*}
	Thus $t=0$ or $t^3+t^2+st+s^2=0$.
	
	If $t=0$, plugging it into (\ref{eq1-2}), we have 
	$x^4+sx=0$ and thus $x=0$ or $x^3=s$. It is clear that $x=0$ is a solution of (\ref{eq1-1}). If $x^3=s=a+b+c\in\gf_{2^m}$, then $x=s^{\frac{1}{3}}\in\gf_{2^m}$ and $y=z=x\in\gf_{2^m}$. Thus $x=x+y+z=t=0$. 
	
	If $t^3+t^2+st+s^2=0$, using $(t_1+1)$ to replace $t$, we  get 
	\begin{equation}
	\label{eq1-5}
	t_1^3+(s+1)t_1+s^2+s=0.
	\end{equation}
	Since 
	\begin{eqnarray*}
		& & \tr_m\left(\frac{(s+1)^3}{(s^2+s)^2}\right) \\
		&=& \tr_m\left( \frac{1}{s} + \frac{1}{s^2} \right) =0 \neq \tr_m(1),
	\end{eqnarray*}
	(\ref{eq1-5})  has exactly one solution in $\gf_{2^m}$ according to Lemma \ref{cubic}. Moreover, we can get the expression of the unique solution by Lemma \ref{cubic_solution}. For the equation $u^2+(s^2+s)u+(s+1)^3=0$, we have $$\left(\frac{u}{s^2+s}\right)^2+\frac{u}{s^2+s}=\frac{1}{s}+\frac{1}{s^2}$$
	and thus $u=s+1$ is a solution and $\epsilon=(s+1)^{\frac{1}{3}}$ is a solution of $x^3=u$ since $\gcd(3,2^m-1)=1$. Furthermore, $r=\epsilon+\frac{a}{\epsilon} = (s+1)^{\frac{1}{3}}+ (s+1)^{\frac{2}{3}} $
	is a solution of (\ref{eq1-5}) and thus $$\bar{t}=r+1=\epsilon+ \epsilon^2 +1,$$ where $\epsilon=(s+1)^{\frac{1}{3}}$, is the unique solution of $t^3+t^2+st+s^2=0$. Namely, $x+y+z$ equals a constant. Plugging $x+y+z=\bar{t}$ into (\ref{eq1-2}), we get
	\begin{equation}
	\label{eq1-6}x^4+(s^2+\bar{t}^2)x+a\bar{t}^2 = 0. 
	\end{equation} 
	Next, using Lemma \ref{quartic}, we will prove that the above equation has two solutions in $\gf_{2^n}$. However, we will also show that the two solutions can not satisfy $x+y+z=\bar{t}$ at the same time and thus (\ref{eq1-2}) has at most one solution in this case. Together with the zero solution, (\ref{eq1-2}) has at most two solutions in $\gf_{2^n}$ and thus $f(x)$ is two-to-one.

	Recall that $\epsilon^3=s+1$ and $\bar{t}=\epsilon+ \epsilon^2 +1$. Since $s^2+\bar{t}^2 = \epsilon^6+\epsilon^4+\epsilon^2$, if $s^2+\bar{t}^2=0$, then $\epsilon=0$ clearly ($\epsilon^2+\epsilon+1\neq0$ due to $m$ odd). Moreover, $s=\epsilon^3+1=1$ and $t=1$. Thus if $s=1$ and $t=1$, (\ref{eq1-6}) becomes $x^4=a$, which has exactly one solution in $\gf_{2^n}$. In the following, we assume that $s^2+\bar{t}^2\neq0$.  Let $f_1(r)=r^3+(s^2+\bar{t}^2)$. Then it is clear that $f_1=(1,2)$, which means that $f_1$ can factor as a product of a linear factor and an irreducible quadratic factor. Moreover, 
	$$r_1 = (s^2+\bar{t}^2)^{\frac{1}{3}} = (s^2+\epsilon^2+\epsilon^4+1)^{\frac{1}{3}} = (\epsilon^2+\epsilon^4+\epsilon^6)^{\frac{1}{3}}
	$$ is the unique solution of $f_1(r)=0$. Set $w_1=a\bar{t}^2\frac{r_1^2
	}{(s^2+\bar{t}^2)^2}$.  In addition,
	\begin{eqnarray*}
		\tr_n(w_1) &=& \tr_n\left( \frac{a(\epsilon+\epsilon^2+1)^2}{(\epsilon^2+\epsilon^4+\epsilon^6)^{\frac{4}{3}}} \right) \\
		&=& \tr_m\left(\tr_{n/m}\left(\frac{a\bar{t}^2}{(\epsilon^2+\epsilon^4+\epsilon^6)^{\frac{4}{3}}}  \right)\right) \\
		&=&\tr_m\left( \frac{s\bar{t}^2}{(\epsilon^2+\epsilon^4+\epsilon^6)^{\frac{4}{3}}} \right)\\
		&=& \tr_m\left( \frac{(\epsilon^3+1)(\epsilon+\epsilon^2+1)^2}{(\epsilon^2+\epsilon^4+\epsilon^6)^{\frac{4}{3}}} \right)\\
		&=&\tr_m\left( \frac{(\epsilon^2+\epsilon+1)^{\frac{4}{3}}}{\epsilon^{\frac{8}{3}}} + \frac{(\epsilon^2+\epsilon+1)^{\frac{1}{3}}}{\epsilon^{\frac{2}{3}}} \right) =0.
	\end{eqnarray*}
	Thus according to Lemma \ref{quartic}, (\ref{eq1-6}) has exactly two solutions in $\gf_{2^n}$, denoted by $x_1, x_2$. Next, we show that the two solutions can not satisfy $x+y+z=\bar{t}$ at the same time. Clearly, there exist some $\alpha,\beta\in\gf_{2^n}$ such that (\ref{eq1-6}) becomes 
	$$\left(x^2+\alpha x+\beta\right) \left(x^2+\alpha x+\alpha^2+\beta \right) = 0$$
	and by comparing the coefficient of $x$, we know that $\alpha^3=(s^2+\bar{t}^2)\neq0$. In addition, by the Vieta theorem, $x_1+x_2=\alpha\neq0$. Thus the two solutions can not satisfy $x+y+z=\bar{t}$ at the same time. 
\end{proof}

\subsection{Binary linear codes from these new two-to-one functions 
}

\begin{table}[!t]
	\centering
	\caption{The weight distribution of the codes $\mathcal{C}_{f}$ in Theorem \ref{Th_n=km}} 	\label{Th_n=km_table1}
	\begin{tabular}{cc}
		\hline
		Weight & Multiplicity \\
		\hline
		$0$ & $1$\\
		$2^{n-1}-2^{\frac{n+m-1}{2}}$ & $2^{n-1}+2^{\frac{n+m}{2}-1} - 2^{n-m-1} - 2^{\frac{n-m}{2}-1}$  \\
		$2^{n-1}$ & $2^{n+m} + 2^{n-m} - 2^n -1$ \\
		$2^{n-1}+2^{\frac{n+m-1}{2}}$ & $2^{\frac{n-m}{2}-1} + 2^{n-1} - 2^{\frac{n+m}{2}-1} - 2^{n-m-1}$\\
		\hline
	\end{tabular}
\end{table}

\begin{table}[!t]
	\centering
	\caption{The weight distribution of the codes $\mathcal{C}_{D(f)}$ in Theorem \ref{Th_n=km}} 	\label{Th_n=km_table2}
	\begin{tabular}{cc}
		\hline
		Weight & Multiplicity \\
		\hline
		$0$ & $1$\\
		$2^{n-2}-2^{\frac{n+m-4}{2}}$ & $2^{n-m-1}+2^{\frac{n-m-2}{2}}$  \\
		$2^{n-2}$ & $2^n- 2^{n-m}-1$ \\
		$2^{n-2}+2^{\frac{n+m-4}{2}}$ & $2^{n-m-1}-2^{\frac{n-m-2}{2}}$\\
		\hline
	\end{tabular}
\end{table}

\begin{Th}\label{Th_n=km}
		Let $n=km$ with $k,m$ odd and $f(x) = \tr_{n/m}\left(x^{2^m+1}\right)+x$. Define two linear codes $\mathcal{C}_{f}$ and $\mathcal{C}_{D(f)}$ as in \eqref{Eq-GenCons1} and \eqref{Eq-GenCons2}, respectively. 
	Then, 
	\begin{enumerate}[(1)]
	\item  $\mathcal{C}_{f}$ is a $\left[2^{n}-1, n+m\right]$ binary linear code with weight distribution in Table \ref{Th_n=km_table1}.
\item  $\mathcal{C}_{D(f)}$ is a $\left[ 2^{n-1}-1,n\right]$ binary linear code with  weight distribution in Table \ref{Th_n=km_table2}
	\end{enumerate}
\end{Th}

\begin{proof}
	Firstly, we shall determine the value $W_f(a,b)$. It is clear that when $(a,b)=(0,0), W_f(a,b)=2^n$.  
	Note that $\tr_n(b \tr_{n/m}(x^{2^m+1}))=\tr_n( \tr_{n/m}(b)x^{2^m+1})$.
	 Thus, if $\tr_{n/m}(b)=0$, then $$W_f(a,b)= \sum_{x\in \gf_{2^n}}(-1)^{\tr_n((a+b) x)},$$ which equals $2^n$ if $a=b$ and $0$ otherwise. 
	 
	 In the following, we consider the case that $\tr_{n/m}(b)\neq0$.	
	Let $\varphi_b(x)=\tr_n(bf(x))$. Since it is quadratic, according to (\ref{Eq-Quad-2}), we need to determine the dimension of kernel of the bilinear form of $\varphi_b(x)$. Note that the bilinear form of $\varphi_b(x)$ is given by 
	\begin{eqnarray*}
		B_{\varphi_b}(x,y) &=& \varphi_b(x+y) + \varphi_b(x) + \varphi_b(y) \\
		&=& \tr_n\left( \tr_{n/m}(b)\left( yx^{2^m} + y^{2^m}x \right)   \right)   = \tr_n\left( L_b(y)x^{2^m} \right),
	\end{eqnarray*}
	where $L_b(y)=  \tr_{n/m}(b)(y+y^{2^{2m}})$. Then $\ker(L_b)=\gf_{2^m}$.  It follows from (\ref{Eq-Quad-2}) that
	$$ W_f(a,b)  = \begin{cases}
	\pm 2^{\frac{n+m}{2}}, & \text{if } \tr_n(bf(x)+ax)=0 \text{ for all } x \in \gf_{2^m}, \\
	0, & \text{otherwise.}  
	\end{cases}$$
Moreover, for $x\in\gf_{2^m}$, 
$$\tr_n\left(bf(x)+ax\right) = \tr_n((a^2+b+b^2)x^2) = \tr_m\left( \tr_{n/m}(a^2+b+b^2) x^2 \right), $$
always equals $0$ if and only if $\tr_{n/m}(a^2+b+b^2)=0$. Thus for any $a,b$ satisfying $\tr_{n/m}(a^2+b^2+b)\neq0$, $W_f(a,b)\in \left\{0, \pm 2^{\frac{n+m}{2}}\right\}.$ Similarly, for any $b$	with $\tr_{n/m}(b) \not\in \{0, 1\}$, $W_f(b)\in \left\{0, \pm 2^{\frac{n+m}{2}}\right\}$. 

\smallskip

Next, we consider the parameters of $\mathcal{C}_f$ and $\mathcal{C}_{D(f)}$, respectively. 

(1) For the linear code $\mathcal{C}_f$, since $W_f(a,b)=2^n$ if and only if  $a=b$ with $\tr_{n/m}(b)=0$, by (\ref{dimensionK1}), the dimension of $K_1 = \left\{ a,b\in\gf_{2^n}: W_f(a,b)=2^n \right\}$ is $n-m$ and thus the dimension of $\mathcal{C}_f$ is $2n-(n-m)=n+m$. Moreover, for any $a,b\in\gf_{2^n}$, $W_f(a,b)\in\left\{ 0, 2^n, \pm 2^{\frac{n+m}{2}} \right\}$. Let 
$$v_1=-2^{\frac{n+m}{2}}, ~~~ v_2 = 0, ~~~ v_3 = 2^{\frac{n+m}{2}}. $$
Then by computing (\ref{Eq-WtDist1}), we can obtain that the occurrences of $W_f(a,b)=v_i$'s, $i=1,2,3$ in the Walsh spectrum of $f$ are 	$$ \begin{cases}
X_1= 2^{\frac{3n-3m}{2}-1} + 2^{2n-m-1} - 2^{\frac{3n-m}{2}-1} - 2^{2n-2m-1} \\ 
X_2 = 2^{2n} + 2^{2n-2m} - 2^{2n-m} -2^{n-m}  \\ 
X_3 = 2^{2n-m-1}+2^{\frac{3n-m}{2}-1} - 2^{2n-2m-1} - 2^{\frac{3n-3m}{2}-1}. \end{cases}
$$
Finally, by (\ref{WtDist-Relation1}), the desired weight distribution of $\mathcal{C}_f$ can be obtained.

(2) For the linear code $\mathcal{C}_{D(f)}$, since $W_f(b)=2^n$ if and only if $b=0$ and for $b\in\gf_{2^n}$, by (\ref{dimension}), the dimension of $\mathcal{C}_{D(f)}$ equals $n$. Moreover, we have
$$W_f(b)\in\left\{ 0,2^n,\pm 2^{\frac{n+m}{2}}  \right\}$$ and by (\ref{wt_new}), the weights of the codewords $\mathbf{c}_b$ in $\mathcal{C}_{D(f)}$ satisfy
$$\mathrm{wt}(\mathbf{c}_b)\in \left\{ 2^{n-2}, 0, 2^{n-2}-2^{\frac{n+m-4}{2}},2^{n-2}+2^{\frac{n+m-4}{2}}  \right\}.$$

 In the following, we determine the weight distribution of $\mathcal{C}_{D(f)}$.  Define 
 $$\mathrm{w}_1 = 2^{n-2}-2^{\frac{n+m-4}{2}},~~ \mathrm{w}_2=2^{n-2},~~ \mathrm{w}_3 = 2^{n-2}+2^{\frac{n+m-4}{2}}.$$
 Then solving (\ref{Eq-WtDist2})  gives the desired weight distribution. 
\end{proof}

\begin{table}[!t]
	\centering
	\caption{The weight distribution of the codes $\mathcal{C}_{D(f)}$ in Theorem \ref{Thn_n=3m}} 	\label{Thn_n=3m_table}
	\begin{tabular}{cc}
		\hline
		Weight & Multiplicity \\
		\hline
		$0$ & $1$\\
		$2^{n-2}-2^{\frac{n+2m-3}{2}}$ & $2^{m-2}+2^{\frac{m-3}{2}}$  \\
		$2^{n-2}$ & $2^n-2^{m-1}-1$ \\
		$2^{n-2}+2^{\frac{n+2m-3}{2}}$ & $2^{m-2}-2^{\frac{m-3}{2}}$\\
		\hline
	\end{tabular}
\end{table}

\begin{Th}\label{Thn_n=3m}
	Let $n=3m$ with $m$ odd and $f(x)=x^{2^{2m+1}+1}+x^{2^{m+1}+1}+x^4+x^3$.  Define two linear codes $\mathcal{C}_{f}$ and $\mathcal{C}_{D(f)}$ as in \eqref{Eq-GenCons1} and \eqref{Eq-GenCons2}, respectively. 
	Then, 
	\begin{enumerate}[(1)]
		\item  $\mathcal{C}_{f}$ is a $\left[2^{n}-1, 2n\right]$ binary linear code with $5$ weights. Moreover, the weights of the codewords $\mathbf{c}_b$ in $\mathcal{C}_{f}$ satisfy
		$$\mathrm{wt}(\mathbf{c}_b)\in\left\{ 2^{n-1}, 0, 2^{n-1}-2^{\frac{n+2m-1}{2}},2^{n-1}+2^{\frac{n+2m-1}{2}}, 2^{n-1} - 2^{\frac{n+m-2}{2}}, 2^{n-1}+ 2^{\frac{n+m-2}{2}} \right\}.$$
		\item  $\mathcal{C}_{D(f)}$ is a $\left[ 2^{n-1}-1,n\right]$ binary linear code with  weight distribution in Table \ref{Thn_n=3m_table}.
	\end{enumerate}
\end{Th}

\begin{proof}
	First of all, we shall determine the value $W_f(a,b)$. It is clear that when $(a,b)=(0,0)$, $W_f(a,b)=2^n$. Let $\varphi_b(x)=\tr_n(bf(x))$.  Since  $f$ is quadratic, according to (\ref{Eq-Quad-2}), we need to determine the dimension of kernel of the bilinear form of $\varphi_b(x)$. Note that the bilinear form of $\varphi_b(x)$ is given by 
	\begin{eqnarray*}
		B_{\varphi_b}(x,y) &=& \varphi_b(x+y) + \varphi_b(x) + \varphi_b(y) \\
		&=& \tr_n\left(  b(x^{2^{2m+1}}y+xy^{2^{2m+1}}+x^{2^{m+1}}y+xy^{2^{m+1}}+x^2y+xy^2) \right)   = \tr_n\left( L_b(y)x^{2} \right),
	\end{eqnarray*}
	where $$L_b(y)=b^2y^{2^{2m+2}}+ b^{2^{2m}}y^{2^{2m}}+b^2y^{2^{m+2}}+b^{2^m}y^{2^m}+b^2y^4+by.$$ Let the dimension of $\ker(L_b)$ be $d_b$.  It follows from (\ref{Eq-Quad-2}) that 
	$$ W_f(a,b)  = \begin{cases}
	\pm 2^{\frac{n+d_b}{2}}, & \text{if } \tr_n(bf(x)+ax)=0 \text{ for all } x \in \gf_{2^m}, \\
	0, & \text{otherwise.}  
	\end{cases}$$ 
	
		Next, we consider the equation $L_b(y)=0$, i.e.,
	$$b^2\tr_{n/m}(y)^4=\tr_{n/m}(by).$$
	Since $\tr_{n/m}(y), \tr_{n/m}(by)\in\gf_{2^m}$, we have $b\in\gf_{2^m}^{*}$ or $\tr_{n/m}(y)=\tr_{n/m}(by)=0$. 
	
	{\bfseries Case 1:} If $b\in\gf_{2^m}^{*}$, then the equation $L_b(y)=0$ becomes $\tr_{n/m}(y)=0 $ or  $\sqrt[3]{b^{-1}}$. Thus in this case, the number of solutions of $L_b(y)$ is $2^{2m+1}$. Namely, $d_b=2m+1$. In the following, we show that there exist some $b\in\gf_{2^m}^{*}$ such that the restriction of $\tr_n(b(f(x)))$ on $\ker(L_b)=\left\{ y\in \gf_{2^n} : \tr_{n/m}(y)=0 ~~\text{or}~~\sqrt[3]{b^{-1}} \right\}$ is the all-zero mapping or not, i.e., $W_f(b)\in\left\{ 0, \pm2^{\frac{n+2m+1}{2}} \right\}$ for $b\in\gf_{2^m}^{*}$. On one hand,  if $\tr_{n/m}(y)=0$, $$\tr_n(bf(y)) = \tr_n\left(b(y\tr_{n/m}(y)^2+y^4)\right)=\tr_n(by^4)= \tr_m\left(b\tr_{n/m}(y^4)\right)=0.$$
On the other hand,	if $\tr_{n/m}(y)=\sqrt[3]{b^{-1}}$, 
	\begin{eqnarray*}
		\tr_n(bf(y)) &=& \tr_n\left(b(y\tr_{n/m}(y)^2+y^4)\right) \\
		&=& \tr_m\left( \tr_{n/m} \left(  by\tr_{n/m}(y)^2 \right) + \tr_{n/m}(by^4) \right) \\
		&=& \tr_m\left( b\tr_{n/m}(y)^3 + b\tr_{n/m}(y)^4 \right) \\
		&=& \tr_m(1+\tr_{n/m}(y)) = 1+\tr_m\left(\sqrt[3]{b^{-1}}\right).
	\end{eqnarray*}
	Then  $\varphi_b(y)=0$ if and only if  $\tr_m\left(\sqrt[3]{b^{-1}}\right)=1$ and  thus  the restriction of $\tr_n(b(f(x)))$ on $\ker(L_b)$ is the all-zero mapping if and only if $\tr_m\left(\sqrt[3]{b^{-1}}\right)=1$. Therefore $W_f(b)\in\left\{ 0, \pm2^{\frac{n+2m+1}{2}} \right\}$ and then clearly $W_f(a, b)\in\left\{ 0, \pm2^{\frac{n+2m+1}{2}} \right\}$ in this case. 
	
	{\bfseries Case 2:} If $b\in\gf_{2^n}\backslash\gf_{2^m}$, then $$\ker(L_b) = \left\{ y: y\in\gf_{2^n} ~~\text{and}~~ \tr_{n/m}(y)=\tr_{n/m}(by)=0 \right\}.$$ For any $b\in\gf_{2^n}\backslash\gf_{2^m}$, define 
	\begin{equation*}
	\begin{cases} 
	b + b^{2^m} + b^{2^{2m}} = \alpha \\ 
	b^{2^m+1}+b^{2^{2m}+1}+b^{2^{2m}+2^m} = \beta  \\ 
	b^{2^{2m}+2^m+1} = \gamma \end{cases}
	\end{equation*}
	and $g(x)=x^3+\alpha x^2+\beta x+\gamma\in\gf_{2^m}[x]$.  Then it is clear that $g(x)=(x+b)\left(x+b^{2^m}\right)\left(x+b^{2^{2m}}\right)$ and $g(x)$ is irreducible on $\gf_{2^m}$. Since $g(x+\alpha) = x^3+(\alpha^2+\beta)x+\alpha\beta+\gamma$ is also irreducible, we have $$\alpha^2+\beta\neq0 ~~\text{and}~~\alpha\beta+\gamma\neq0.$$
	In addition, for any fixed $b\in\gf_{2^n}\backslash\gf_{2^m},$ it is well known that for any $y\in\gf_{2^n}$, there exist unique $y_0,y_1,y_2\in\gf_{2^m}$ such that $y=y_0+y_1b+y_2b^2$. Then  $$\tr_{n/m}(y)=y_0+y_1\alpha+y_2\alpha^2=0$$
	and 
	\begin{eqnarray*}
		\tr_{n/m}(by) &=& \tr_{n/m}\left( y_0b+y_1b^2+y_2b^3 \right) \\
		&=& \tr_{n/m}\left( (y_1+y_2\alpha)b^2 + (y_0+y_2\beta) b +y_2\gamma  \right) \\
		&=& \alpha y_0 + \alpha^2y_1 + \left(\gamma+\alpha\beta+\alpha^3\right)y_2 = 0.
	\end{eqnarray*}
	Plugging $y_0=y_1\alpha+y_2\alpha^2$ into the above equation and simplifying it, we obtain $(\gamma+\alpha\beta)y_2=0$ and then $y_2=0$ since $\gamma+\alpha\beta\neq0$. Thus $$\ker(L_b) = \left\{ (\alpha+b)\eta: \eta\in\gf_{2^m} \right\}.$$ Clearly, in this case, the dimension of $\ker(L_b)$ is $m$. Moreover, for $x\in\ker(L_b),$ 
	\begin{eqnarray*}
	\tr_n(bf(x)+ax) &=& \tr_n\left(b(x\tr_{n/m}(x)^2+x^4) +ax \right)=\tr_n((b+a^4)x^4) \\
	&=& \tr_n\left( (b+a^4) (\alpha+b)^4\eta^4 \right) \\
	&=& \tr_m\left( \tr_{n/m}\left( (b+a^4)(\alpha+b)^4 \right)\eta^4 \right).
	\end{eqnarray*}
Obviously, if $a^4=b$, the restriction of $\tr_n(b(f(x))+ax)$ on $\ker(L_b)$ is the all-zero mapping and  thus $W_f(a,b)=\pm 2^{\frac{n+m}{2}}$. 

Moreover, if $a=0$, $\tr_n(bf(x))=\tr_m\left(\tr_{n/m}\left(b(\alpha+b)^4\right)\eta^4\right)=\tr_m\left( U_b \eta^4  \right),$ where $$U_b = \tr_{n/m}(b(\alpha^4+b^4)) = \alpha^5 + b^5+b^{5\cdot 2^m} + b^{5\cdot 2^{2m}}.$$
In the following, we will show that $U_b\neq0$ for any $b\in\gf_{2^n} \backslash\gf_{2^m}$. If there exist some $b\in\gf_{2^n}\backslash\gf_{2^m}$ such that $U_b=0$, then by simplifying it, we get
$$(b+b^{2^m})(b+b^{2^{2m}})(b^{2^m}+b^{2^{2m}})\left(\alpha^2+\beta\right)=0,$$
which is impossible since $b\not\in\gf_{2^m}$ and $\alpha^2+\beta\neq0$. Thus for any $b\in\gf_{2^n}\backslash\gf_{2^m}$, $U_b\neq0$  and then the restriction of $\tr_n(b(f(x)))$ on $\ker(L_b)$ can not be the all-zero mapping. Thus $W_f(b)=0$. 

In conclusion, for any $a,b\in\gf_{2^n}$, $$W_f(a,b)\in\left\{ 0, 2^n, \pm2^{\frac{n+2m+1}{2}}, \pm  2^{\frac{n+m}{2}}  \right\}.$$ However, for any $b\in\gf_{2^n}$, $$W_f(b) \in\left\{ 0, 2^n, \pm2^{\frac{n+2m+1}{2}}  \right\}.$$

Next, we consider the parameters of $\mathcal{C}_f$ and $\mathcal{C}_{D(f)}$, respectively. 

(1) For the linear code $\mathcal{C}_f$, since $W_f(a,b)=2^n$ if and only if $(a,b)=(0,0)$, by (\ref{dimensionK1}), the dimension of $\mathcal{C}_f$  is $2n$. Moreover, since for any $a,b\in\gf_{2^n}$, $W_f(a,b)\in\left\{ 0, 2^n, \pm2^{\frac{n+2m+1}{2}}, \pm  2^{\frac{n+m}{2}}  \right\},$ by (\ref{WtDist-Relation1}), the weights of the codewords $\mathbf{c}_b$ in $\mathcal{C}_{f}$ satisfy
$$\mathrm{wt}(\mathbf{c}_b)\in\left\{ 2^{n-1}, 0, 2^{n-1}-2^{\frac{n+2m-1}{2}},2^{n-1}+2^{\frac{n+2m-1}{2}}, 2^{n-1} - 2^{\frac{n+m-2}{2}}, 2^{n-1}+ 2^{\frac{n+m-2}{2}} \right\}.$$

(2) For the linear code $\mathcal{C}_{D(f)}$, $W_f(b)= 2^n$ if and only if $b=0$, which means that the dimension of $\mathcal{C}_{D(f)}$ is $n$ according to (\ref{dimension}). Since for any $b\in\gf_{2^n}$, $W_f(b)\in\left\{ 0, 2^n, \pm2^{\frac{n+2m+1}{2}} \right\}$, by (\ref{wt_new}), the weights of the codewords $\mathbf{c}_b$ in $\mathcal{C}_{D(f)}$ satisfy
$$\mathrm{wt}(\mathbf{c}_b)\in \left\{ 2^{n-2}, 0, 2^{n-2}- 2^{\frac{n+2m-3}{2}},2^{n-2}+2^{\frac{n+2m-3}{2}} \right\}.$$

In the following, we determine the weight distribution of $\mathcal{C}_{D(f)}$.  Define 
$$\mathrm{w}_1 = 2^{n-2}-2^{\frac{n+2m-3}{2}},~~ \mathrm{w}_2=2^{n-2},~~ \mathrm{w}_3 = 2^{n-2}+2^{\frac{n+2m-3}{2}}.$$
Then solving (\ref{Eq-WtDist2})  gives the desired weight distribution. 
\end{proof}

\begin{Rem}
\emph{	In Theorems \ref{Th_n=2m+1_3} - \ref{Th_n=3m}, the linear codes $\mathcal{C}_{D(f)}$ has the same parameters as the shortened Hadamard codes,
	which is locally decodable code that provides a way to recover parts of the original message with high probability, while only looking at a small fraction of the received word.
	This property gives rise to applications in the computational complexity theory and in the CDMA communication system.
	The dual codes of $\mathcal{C}_{D(f)}$ are the binary Hamming codes with parameters $[2^{n-1}-1, 2^{n-1}-n, 3]$.}
\end{Rem}

\begin{Rem}
\emph{	In \cite{lirecent} there are several other classes of two-to-one quadratic polynomials. {The experiment results show that we can obtain 3-weight or 5-weight binary linear codes as well from  generalized quadratic polynomials. 
	Due to the similarities of the parameters of those codes and the proofs, we choose some representatives of them that are more difficult and omitted the others in this paper.} In addition, the linear code $\mathcal{C}_f$ in Theorem \ref{Th_n=2m+1_2}
	appears to be a 3-weight code by numerical results. Nevertheless, we didn't manage to prove it by the techniques used in this paper.
	We cordially invite interested readers to determine the weight distribution of the linear codes $\mathcal{C}_f$ in Theorem \ref{Th_n=2m+1_2} and Theorem \ref{Thn_n=3m}.}
\end{Rem}

{
\begin{Prob}
	Determine the weight distribution of the linear codes $\mathcal{C}_f$ in Theorem \ref{Th_n=2m+1_2} and Theorem \ref{Thn_n=3m}.
\end{Prob}	
	
	According to the experiment results, we also have the following conjecture. 

\begin{Conj}
	Let $n=2m+1$ and $f(x)=x^{3\cdot 2^{m+1}} + x^{2^{m+2}+1}+x^{2^{m+1}+1}+x$. Then $f(x)$ is two-to-one over $\gf_{2^n}$. Moreover, when $m\ge4$, the linear code $\mathcal{C}_{D(f)}$ has the parameters $\left[ 2^{n-1}-1,n,  2^{n-1}-2^{\frac{n-1}{2}}\right]$ and the weights of the codewords $\mathbf{c}_b$ in $\mathcal{C}_{D(f)}$ satisfy
	$$\mathrm{wt}(\mathbf{c}_b)\in\left\{ 2^{n-2}, 0, 2^{n-2}-2^{\frac{n-3}{2}},2^{n-2}+2^{\frac{n-3}{2}}, 2^{n-2}-2^{\frac{n-1}{2}},2^{n-2}+2^{\frac{n-1}{2}} \right\}.$$ If possible, determine the weight distribution of the linear code $\mathcal{C}_{D(f)}$. 
\end{Conj}}

{\section{Binary linear codes from $(x^{2^t}+x)^e$}

It is clear that the function $\left(x^{2^t}+x\right)^e$ with $\gcd(t, n)=1$ and $\gcd(e, 2^n-1)$ is a two-to-one function from $\gf_{2^n}$ to itself.
In this section, we construct binary linear codes from two-to-one functions in this form. 

Recall that given a two-to-one function $f$, the parameters of the linear codes $\mathcal{C}_{D(f)}$ in \eqref{Eq-GenCons2} depend on the investigation of the value $W_f(b)$. 
We first present an interesting relation on $W_f(b)$ for $f(x)=(x^{2^t}+x)^e$ and the Walsh transform of $\tr_n(x^e)$. Actually, we consider a general form $f(x)=P(\psi(x))$, where $P$ is a permutation polynomial over $\gf_{2^n}$ and $\psi(x)$ is two-to-one with $\mathrm{Im}(\psi) = \{y \in \gf_{2^n} : \tr_n(y)=0\}$.  

\begin{Prop}
	\label{relation}
	Let $f(x)=P(\psi(x))$, where $P$ is a permutation polynomial over $\gf_{2^n}$ and $\psi(x)$ is two-to-one with $\mathrm{Im}(\psi) = \{y \in \gf_{2^n} : \tr_n(y)=0\}$.
    Then for any $b\in\gf_{2^n}^{*}$, 
    $$ W_f(b)=\sum_{x\in \gf_{2^n}}(-1)^{\tr_n\left(bP(y)+y\right)}.$$
\end{Prop} 

\begin{proof}
	Let  
	$$
	T_0= \{y \in \gf_{2^n} : \tr_n(y)=0\}.
	$$ Take an element $a\in \gf_{2^n}$ with $\tr_n(a)=1$. Then 
	$$
	T_1 := \{y \in \gf_{2^n} : \tr_n(y)=1\} = \{a+y\,:\, y \in T_0\}.
	$$
	For any $b\in\gf_{2^n}^{*}$, the fact $\sum_{y\in \gf_{2^n}} (-1)^{\tr_n\left(bP(y)\right)}=0$ implies 
	$$
	\sum_{y\in T_0} (-1)^{\tr_n\left(bP(y)\right)}  = - \sum_{y\in T_1} (-1)^{\tr_n\left(bP(y)\right)} = \sum_{y\in T_0} (-1)^{\tr_n\left(bP(y+a)+a\right)}.
	$$
	Thus, 
	\begin{eqnarray*}	
		W_f(b) &=& \sum_{x\in \gf_{2^n}}(-1)^{\tr_n\left(bP(\psi(x))\right)} 
		          = 2 \sum_{y\in T_0} (-1)^{\tr_n\left(bP(y)\right)}   
		\\     &=& \sum_{y\in T_0} (-1)^{\tr_n\left(bP(y)\right)} + \sum_{y \in T_0}(-1)^{\tr_n\left( bP(y+a)+a\right)} 
		\\     &=& \sum_{y\in T_0}(-1)^{\tr_n\left(bP(y)+y\right)} + \sum_{y\in T_0}(-1)^{\tr_n\left(bP(y+a)+y+a\right)} 		
		\\     &=& \sum_{y\in T_0}(-1)^{\tr_n\left(bP(y)+y\right)} + \sum_{y\in T_1}(-1)^{\tr_n\left(bP(y)+y\right)} 		
		\\     &=& \sum_{y\in \gf_{2^n}}(-1)^{\tr_n\left(bP(y)+y\right)}. 
	\end{eqnarray*}
\end{proof}

\begin{Rem}
\emph{	It is well known that given an integer $e$ with $\gcd(e, 2^n-1)=1$, the calculations of the weight distribution of $\mathcal{C}_f$ with $f(x)=x^e$, 
	the Walsh spectrum of $x^e$, the cross-correlation distribution of $m$-sequences and their $e$-decimated sequences are equivalent.
	The relation has provided a great amount of interesting results which originated from cryptography, coding theory and sequence design.
	Proposition \ref{relation} exhibits a similar relation, which indicates the equivalence between the computation of the weight distribution of $C_{D(f)}$ for $f(x)=(x^{2^t}+x)^e$ and 
	the Walsh spectrum of $\tr_n(x^e)$. 
	In other words, any power function $x^e$, $\gcd(e, 2^n-1)=1$, with $t$-valued Walsh spectrum
	can be employed to construct linear codes $\mathcal{C}_{D(f)}$ with $t$ nonzero weights.
}

\emph{	Recently Li and Zeng in \cite{Li2018} surveyed the exponents $e$ that allow for 3-valued, 4-valued, 5-valued Walsh spectra of $x^e$.
	All the exponents $e$ listed in \cite{Li2018} with $\gcd(e, 2^n-1)=1$ can be employed to generate binary linear codes $C_{D(f)}$ with few weights.}
\end{Rem}

\begin{table}[!t]
	\centering
	\caption{Known almost bent power functions $x^e$ over $\gf_{2^n}$, $n$ odd} \label{AB}
	\begin{tabular}{cccc}
		\hline
	Functions & $e$ & Conditions & References \\
		\hline
		Gold & $2^i+1$ & $\gcd(i,n)=1$ &  \cite{gold1968maximal,nyberg1993differentially} \\
		Kasami & $2^{2i}-2^i+1$ & $\gcd(i,n)=1$ &  \cite{kasami1971weight}			 \\
		Welch & $ 2^m+3$ & $n=2m+1$ & \cite{canteaut2000binary,hollmann2001proof} \\
	Niho-1 & $2^m+2^{\frac{m}{2}}-1$ & $n=2m+1$, $m$ even & \cite{hollmann2001proof} \\
	Niho-2 & $2^m+2^{\frac{3m+1}{2}}-1$ & $n=2m+1$, $m$ odd & \cite{hollmann2001proof} \\
		\hline
	\end{tabular}
\end{table}

	For simplicity, we only provide the result from almost bent functions over $\gf_{2^n}$ with $n$ odd, which 
has three-valued Walsh spectrum $\left\{ 0, \pm 2^{\frac{n+1}{2}} \right\}$ \cite{chabaud1994links}.  
The known almost bent exponents $e$ is listed in Table \ref{AB}.
From Proposition \ref{relation}, we have the following theorem on the linear codes $\mathcal{C}_{D(f)}$ defined as in (\ref{Eq-GenCons2}).

\begin{Th}
	\label{AB_codes}
	Let $n=2m+1$, $f(x)=\left(x^{2^t}+x\right)^e$ with $\gcd(t, n)=1$ and $e$ being one of the almost bent exponents in Table \ref{AB}. Let $\mathcal{C}_{D(f)}$ is defined as in \eqref{Eq-GenCons2}. Then
		  $\mathcal{C}_{D(f)}$ is a $\left[ 2^{n-1}-1,n\right]$ binary linear code with weight distribution in Table \ref{Th_n=2m+1_table4}.
\end{Th}

\begin{table}[!t]
	\centering
	\caption{The weight distribution of the codes $\mathcal{C}_{D(f)}$ in Theorem \ref{AB_codes}} \label{Th_n=2m+1_table4}
	\begin{tabular}{cc}
		\hline
		Weight & Multiplicity \\
		\hline
		$0$ & $1$\\
		$2^{n-2}-2^{m-1}$ & $2^{n-2}+2^{m-1}$  \\
		$2^{n-2}$ & $ 2^{n-1}-1$ \\
		$2^{n-2}+2^{m-1}$ & $2^{n-2}-2^{m-1}$\\
		\hline
	\end{tabular}
\end{table}

\begin{proof}
	From Proposition \ref{relation} and the almost bent property of $x^e$, we know that for $b\in\gf_{2^n}^{*}$, $$W_f(b)\in\left\{ 0,  \pm 2^{\frac{n+1}{2}} \right\}.$$
	Moreover, it is clear that $W_f(b)= 2^n$ if and only if $b=0$, which means that the dimension of $\mathcal{C}_{D(f)}$ is $n$ according to (\ref{dimension}). Furthermore,  by (\ref{wt_new}), the weights of the codewords $\mathbf{c}_b$ in $\mathcal{C}_{D(f)}$ satisfy
	$$\mathrm{wt}(\mathbf{c}_b)\in \left\{ 2^{n-2}, 0, 2^{n-2}- 2^{\frac{n-3}{2}},2^{n-2}+2^{\frac{n-3}{2}} \right\}.$$
Finally,  define 
	$$\mathrm{w}_1 = 2^{n-2}-2^{\frac{n-3}{2}},~~ \mathrm{w}_2=2^{n-2},~~ \mathrm{w}_3 = 2^{n-2}+2^{\frac{n-3}{2}}.$$
	Then solving (\ref{Eq-WtDist2})  gives the desired weight distribution. 
\end{proof}

As for the linear codes $\mathcal{C}_f$ defined as in \eqref{Eq-GenCons1}, where $f=(x^{2}+x)^e$ for examples, it seems hard to compute the Walsh transform $W_f(a,b)$ for any $a,b\in\gf_{2^n}$. 
However, for the Gold function, we manage to determine its possible values. 

\begin{Th}
	\label{Th_n=2m+1_5}
	Let $n=2m+1$ and  $i$ be a positive integer with $\gcd(i,n)=1$. Let $f(x)=\left(x^{2^t}+x\right)^{2^i+1}$ with $\gcd(t,n)=1$. Define the linear code $\mathcal{C}_f$ as in \eqref{Eq-GenCons1}. Then,
	 $\mathcal{C}_f$ is a $\left[ 2^n-1, 2n \right]$ binary code with five weights. Moreover, the weights of the codewords $\mathbf{c}_b$ in $\mathcal{C}_f$ satisfy
		$$\mathrm{wt}(\mathbf{c}_b) \in \left\{ 2^{n-1}, 0, 2^{n-1} - 2^{\frac{n-1}{2}},   2^{n-1} + 2^{\frac{n-1}{2}},  2^{n-1} - 2^{\frac{n+1}{2}}, 2^{n-1} + 2^{\frac{n+1}{2}}    \right\}. $$  
\end{Th}

\begin{proof}
	First of all, we shall compute the value $W_f(a,b)$. It is clear that when $(a,b)=(0,0)$, $W_f(a,b)=2^n$. For $(a,b)\neq (0,0)$, let $\varphi_b(x) = \tr_n(b(f(x)))$. Since $f$ is quadratic, from \eqref{Eq-Quad-2}, we need to compute the dimension of kernel of the bilinear form of $\varphi_{b}(x)$. Note that the bilinear form of $\varphi_{b}(x)$ is given by 
	\begin{eqnarray*}
	B_{\varphi_{b}}(x,y) &=& \varphi_{b}(x+y) + \varphi_{b}(x)+\varphi_{b}(y) \\
	&=& \tr_n\left( b\left( y^{2^t}x^{2^{i+t}} + y^{2^{i+t}}x^{2^t} + yx^{2^{i+t}} + y^{2^{i+t}} x + y^{2^t}x^{2^i} + y^{2^i} x^{2^t} + yx^{2^i} + y^{2^i}x  \right)  \right) \\
	&=& \tr_n\left( L_b(y) x^{2^{i+t}} \right),
	\end{eqnarray*}
	where $$L_b(y)=\left( b^{2^t}y^{2^{i+2t}} + (b+b^{2^t}) y^{2^{i+t}} + by^{2^i}\right)^{2^i} + b^{2^t}y^{2^{2t}} + (b+b^{2^t})y^{2^t}+by.$$
	Let $\ker(L_b) = \left\{ y: y \in \gf_{2^n} ~~\text{and}~~L_b(y) = 0 \right\}.$ Next we determine the dimension of $\ker(L_b)$. Let $y^{2^t}+y=z$. Then 
	$$L_b = b^{2^{i+t}}z^{2^{2i+t}}+b^{2^i}z^{2^{2i}}+b^{2^t}z^{2^t} + bz =0,$$
	which means $b^{2^i}z^{2^{2i}}+bz =0$ or $1$. From $b^{2^i}z^{2^{2i}}+bz =0$, since $\gcd\left( i, n \right)=1$, we have $z=b^{-\frac{1}{2^i+1}}$. Thus $L_b(y)=0$ has  at most eight solutions in $\gf_{2^n}$, namely, the dimension of $\ker(L_b)$ is at most $3$.  Moreover, since $W_f(a,b) \in\left\{ 0, \pm 2^{\frac{n+d_b}{2}} \right\}$ and $n$ is odd, where $d_b$ is the dimension of $\ker(L_b)$, $d_b= 1$ or $3$. Therefore, for any $a, b\in\gf_{2^n}$, 
	$$W_f(a,b) = \left\{ 0, 2^n, \pm 2^{\frac{n+1}{2}}, \pm 2^{\frac{n+3}{2}} \right\}.  $$
	
	Next, we consider the parameters of $\mathcal{C}_f$. Since $W_f(a,b)=2^n$ if and only if $(a,b)=(0,0)$, by (\ref{dimensionK1}), the dimension of $\mathcal{C}_f$  is $2n$. Moreover, since for any $a,b\in\gf_{2^n}$, $W_f(a,b)\in\left\{ 0, 2^n, \pm 2^{\frac{n+1}{2}}, \pm 2^{\frac{n+3}{2}} \right\},$ by (\ref{WtDist-Relation1}), the weights of the codewords $\mathbf{c}_b$ in $\mathcal{C}_{f}$ satisfy
	$$\mathrm{wt}(\mathbf{c}_b)\in\left\{ 2^{n-1}, 0, 2^{n-1} - 2^{\frac{n-1}{2}},   2^{n-1} + 2^{\frac{n-1}{2}},  2^{n-1} - 2^{\frac{n+1}{2}}, 2^{n-1} + 2^{\frac{n+1}{2}}    \right\}.$$
\end{proof}

Moreover, according to the experiment results, we have the following conjecture. 
\begin{Conj}
	Let $n=2m+1$ and $e$ be the almost bent exponents as given in Table \ref{AB}. Let $f(x)=(x^{2^t}+x)^e$ with $\gcd(t,n)=1$. Define the linear codes $\mathcal{C}_f$  as in \eqref{Eq-GenCons1}. Then, the parameters of the linear codes $\mathcal{C}_f$ are the same as that in Theorem \ref{Th_n=2m+1_5}. If possible, determine the weight distribution of the linear codes $\mathcal{C}_f$.
\end{Conj}
}

\section{Conclusion}
In this paper, we employed {some known and new} two-to-one functions  in two generic constructions of binary linear codes.
By investigating the Walsh transform of relevant quadratic functions, we obtained the possible Hamming weights of the codewords in constructed linear codes.
The two-to-one functions with few-valued Walsh transforms are particularly studied. As a result, a large number of new binary codes with few weights are presented.
Moreover, the weight distributions of the codes with one nonzero weight and with three nonzero weights are determined.

\bibliographystyle{plain}
\bibliography{ref}

\end{document}